\documentclass[a4paper,UKenglish,cleveref, autoref]{lipics-v2019}

\usepackage{xargs}
\usepackage[labelformat=simple]{subcaption} 

\usepackage{mdframed}  
\usepackage{todonotes} 

\graphicspath{{fig/}{./fig/}}

\AtBeginDocument{%
  \maketitle

  {\footnotetext{
      \begin{minipage}[l]{0.2\textwidth}
        \includegraphics[trim=10cm 6cm 10cm
        5cm,clip,scale=0.15]{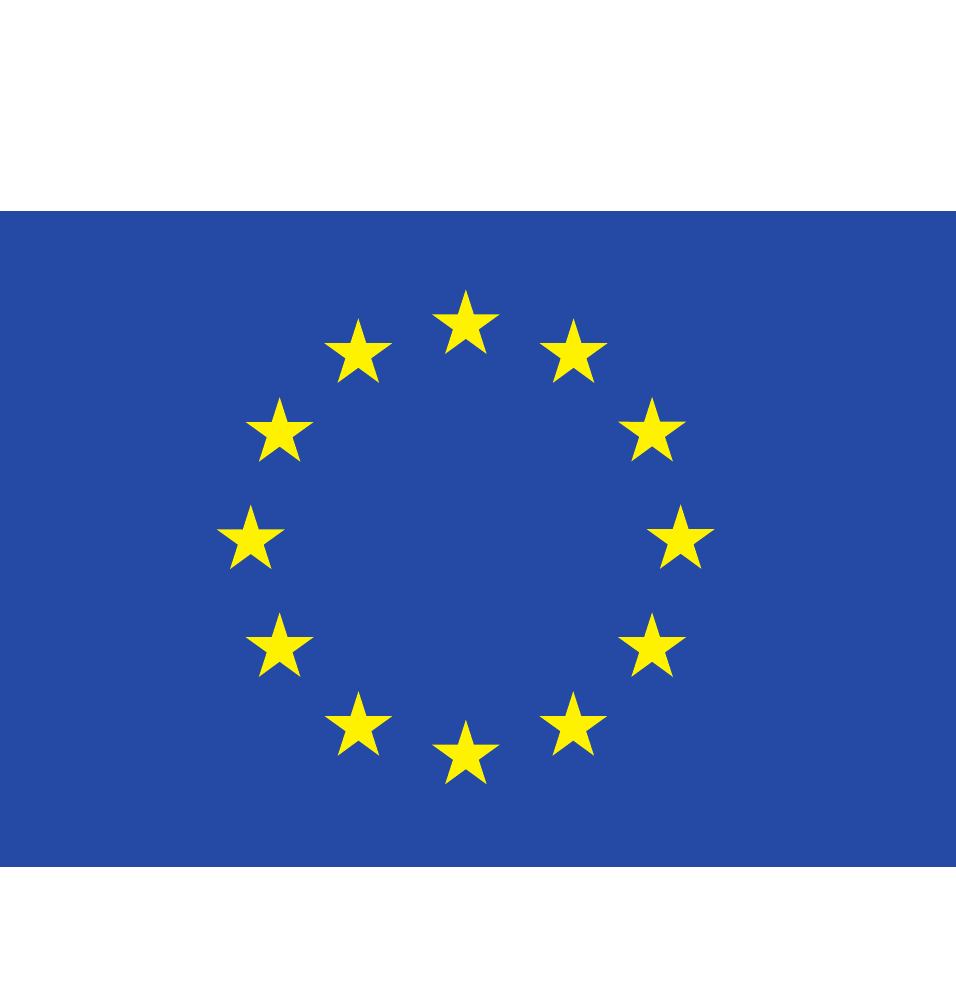}
      \end{minipage}
      \hspace{-1.3cm}
      \begin{minipage}[l][0.8cm]{0.86\textwidth}
        This project has received funding from the European Union's
        Horizon 2020 research and innovation programme under the Marie
        Sk\l{}odowska-Curie grant agreement No 734922.
      \end{minipage}%
  }}
}

\theoremstyle{plain}
\newtheorem{observation}{Observation}

\theoremstyle{definition}
\newtheorem{problem}{Problem}
\newtheorem*{problem*}{Problem}

\newcommand{\sfrac} [2] {{#1}/{#2}}

\newcommand{\R}{\mathbb{R}^{2}}

\renewcommand{\o}{\mathcal{O}}

\newcommand{\ch} [1][P] {\mathcal{CH}({#1})}

\newcommand{\rch} [1][P] {\mathcal{RH}({#1})}

\newcommand{\rcht} [1][P] {\mathcal{RH}_{\theta}({#1})}
\newcommand{\rchtB} [1][B] {\mathcal{RH}_{\theta}({#1})}
\newcommand{\rchtR} [1][R] {\mathcal{RH}_{\theta}({#1})}

\newcommand{\ob}{\o_\beta}
\newcommand{\obh} [1][P] {\o_\beta \mathcal{H}({#1})}

\newcommand{\oh} [1][P] {\o\mathcal{H}({#1})}

\title{Separating bichromatic point sets in the plane by restricted orientation convex hulls\footnote{A preliminary version of this paper was presented at the XVIII Spanish Meeting on Computational Geometry (EGC2019)~\cite{alegria_2019}.}}

\titlerunning{Separating bichromatic point sets by restricted
  orientation convex hulls.}

\author{Carlos Alegr\'{i}a}%
{Dipartimento di Ingegneria, Universit\`{a} Roma Tre}%
{carlos.alegria@uniroma3.it}%
{https://orcid.org/0000-0001-5512-5298}%
{Research supported by MIUR Proj. ``AHeAD'' n\textsuperscript{o}
  20174LF3T8.}

\author{David Orden}%
{Departamento de F\'{\i}sica y Matem\'aticas, Universidad de Alcal\'a,
  Spain}%
{david.orden@uah.es}%
{https://orcid.org/0000-0001-5403-8467}%
{Research supported by Project PID2019-104129GB-I00 / AEI / 10.13039/501100011033 of the Spanish Ministry of Science and Innovation.}

\author{Carlos Seara}%
{Departament de Matem\`{a}tiques, Universitat Polit\`{e}cnica de
  Catalunya, Spain}%
{carlos.seara@upc.edu}%
{https://orcid.org/0000-0002-0095-1725}%
{Research supported by Project PID2019-104129GB-I00 / AEI / 10.13039/501100011033 of the Spanish Ministry of Science and Innovation.}
\author{Jorge Urrutia}%
{Instituto de Matem\'{a}ticas, Universidad Nacional Aut\'{o}noma de
  M\'{e}xico}%
{urrutia@matem.unam.mx}%
{https://orcid.org/0000-0002-4158-5979}%
{Research supported in part by SEP-CONACYT 80268, PAPPIIT IN102117
  Programa de Apoyo a la Investigaci\'on e Innovaci\'on Tecnol\'ogica
  UNAM.}

\authorrunning{C. Alegr\'{i}a et al.}
\Copyright{C. Alegr\'{i}a et al.}



\nolinenumbers 

\hideLIPIcs  

\EventEditors{John Q. Open and Joan R. Access}
\EventNoEds{2}
\EventLongTitle{42nd Conference on Very Important Topics (CVIT 2016)}
\EventShortTitle{CVIT 2016}
\EventAcronym{CVIT}
\EventYear{2016}
\EventDate{December 24--27, 2016}
\EventLocation{Little Whinging, United Kingdom}
\EventLogo{}
\SeriesVolume{42}
\ArticleNo{23}

\begin{document}

\begin{abstract}
  We explore the separability of point sets in the plane by a \emph{restricted-orientation convex hull}, which is an orientation-dependent, possibly disconnected, and non-convex enclosing shape that generalizes the convex hull.
  Let $R$ and $B$ be two disjoint sets of red and blue points in the plane, and $\o$ be a set of $k\geq 2$ lines passing through the origin.
  We study the problem of computing the set of orientations of the lines of $\o$ for which the $\o$-convex hull of $R$ contains no points of $B$.

  For $k=2$ orthogonal lines we have the \emph{rectilinear convex hull}.
  In optimal $O(n\log n)$ time and $O(n)$ space, $n = \vert R \vert + \vert B \vert$, we compute the set of rotation angles such that, after simultaneously rotating the lines of $\o$ around the origin in the same direction, the rectilinear convex hull of $R$ contains no points of $B$.
  We generalize this result to the case where $\o$ is formed by $k \geq 2$ lines with arbitrary orientations.
  In the counter-clockwise circular order of the lines of $\o$, let $\alpha_i$ be the angle required to clockwise rotate the $i$th line so it
  coincides with its successor.
  We solve the problem in this case in $O(\sfrac{1}{\Theta}\cdot N \log N)$ time and $O(\sfrac{1}{\Theta}\cdot N)$ space, where $\Theta = \min \{ \alpha_1,\ldots,\alpha_k \}$ and $N=\max\{k,\vert R \vert + \vert B \vert \}$.
  We finally consider the case in which $\o$ is formed by $k=2$ lines, one of the lines is fixed, and the second line rotates by an angle that goes from $0$ to $\pi$.
  We show that this last case can also be solved in optimal $O(n\log n)$ time and $O(n)$ space, where $n = \vert R \vert + \vert B \vert$.
\end{abstract}

\section{Introduction}\label{sec:intro}

A classic topic in computational geometry is designing efficient algorithms to separate sets of red and blue points.
Several separability criteria have been considered in the literature, as well as separators of different complexities.
Well-known constant-complexity separators include a line or a hyperplane~\cite{aronov_2012,houle_1993,hurtado_2005,megiddo_1983}, a wedge or a double-wedge~\cite{abello_1998,Agarwal_2006,hurtado_2004,hurtado-2001,seara_2002},
a circle~\cite{barba_2012,arkin_2006,boissonnat_2001,megiddo_1986}, and one or two boxes~\cite{acharyya_2019,cortes_2009,moslehi_2016,van_kreveld_2011}.
Typical separators of linear complexity include different types of polygonal chains (e.g. monotone or with alternating constant turn)~\cite{hurtado_2004,pelaez_2011}, different types of enclosing shapes (e.g. a polygon or a non-traditional convex hull)~\cite{alegria_2018,edelsbrunner_1988}, and sets of geometric objects of the same type, such as a hyperplanes~\cite{nakayama_1998} and triangles~\cite{moslehi_2017}.
These choices of separators have been used not only on points, but also
on segments, circles, simple polygons, etc.

Separability problems are closely related to clustering applications, where separating/discriminating is a necessary task.
Consider for example a damaged region modeled by a set of points that needs to be separated from the rest.
In this context, we want to extract the region with minimum area bounded by an enclosing shape that is easy to cut and compute; see~\cite{alegria_2018,alegria_2020,bae_2009,biswas_2012,fink_2004,franek_2008,van_kreveld_2009}
for references on these types of shapes.

In this paper we extend the previous work on separability of two-colored point sets in the Euclidean plane.
Let $P$ be a finite set of points.
The convex hull of $P$, that we denote with $\ch$,
is the closed region obtained by removing from the plane all the open halfplanes which are empty of points of $P$.
We explore the separability by orientation-dependent, possibly disconnected, and non-convex enclosing shapes that generalize this definition by using open wedges instead of open halfplanes.
We first study the \emph{rectilinear convex hull}.
The rectilinear convex hull of $P$, that we denote with $\rch$, 
is the closed region obtained by removing from the plane all the open axis-aligned wedges of aperture angle~$\frac{\pi}{2}$,
which are empty of points of~$P$ (see Section~\ref{sec:rch} for a formal definition).
Observe in \cref{fig:rch} that $\rch$ might be a simply connected set, yielding an intuitive and appealing structure.
However, in other cases $\rch$ can have several connected components, some of which might be single points of $P$.

\begin{figure}[ht]
  \centering%

  \subcaptionbox
  {\label{fig:rch:1}
    The rectilinear convex hull is formed by a single connected component.
  }
  [.4\linewidth][c]{\includegraphics[page=1,scale=0.9]{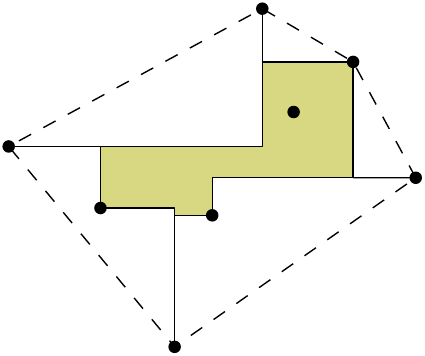}}%
  \qquad%
  \subcaptionbox
  {\label{fig:rch:2}
    The rectilinear convex hull is formed by four connected components, two of which are points of $P$.
  }
  [.4\linewidth][c]{\includegraphics[page=2,scale=0.9]{rch}}%

  \caption
  {
    The rectilinear convex hull of a finite point set $P$.
    The (standard) convex hull of $P$ is shown in dashed lines.
  }
  \label{fig:rch}
\end{figure}

The rectilinear convex hull introduces two important differences with respect to the convex hull.
On one hand we have that $\rch\subset \ch$~\cite[Theorem~4.7]{preparata_1985}, a property that provides more flexibility to better classify a subset of points.
On the other hand we have that $\rch$ is orientation-dependent, which introduces the orientation of the empty wedges as a search space for several optimization criteria; e.g., minimum area or boundary points.
To illustrate these differences, consider two disjoint sets $R$ and $B$ of red and blue points in the plane. Using the standard convex hull, the relative positions of~$R$ and~$B$ may lead to situations as in Figure~\ref{fig:bichromatic_inclusion_ch}.

\begin{figure}[ht]
  \centering
  \subcaptionbox
  {\label{fig:bichromatic_inclusion_ch:1}Convex hull inclusion.}
  [.25\linewidth][c]
  {\includegraphics[page=1,scale=0.9]{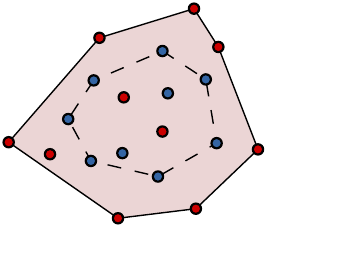}}
  \quad
  \subcaptionbox
  {\label{fig:bichromatic_inclusion_ch:2}Partial containment.}
  [.25\linewidth][c]
  {\includegraphics[page=2,scale=0.9]{bichromatic_inclusion_ch}}
  \quad
  \subcaptionbox
  {\label{fig:bichromatic_inclusion_ch:3}The boundary of $\ch[R]$ separates $R$ from $B$.}
  [.26\linewidth][c]
  {\includegraphics[page=3,scale=0.9]{bichromatic_inclusion_ch}}

  \caption
  {
    Relative positions of $B$ and $\ch[R]$.
    The hull $\ch[B]$ is shown in dashed lines.
  }
  \label{fig:bichromatic_inclusion_ch}
\end{figure}

Using instead the rectilinear convex hull with an arbitrary orientation, we can achieve further goals such as completely separating $R$ and $B$, or minimizing the number of \emph{misclassified points}; i.e., points of one color inside the hull of the other color. See Figure~\ref{fig:bichromatic_inclusion_rch}.

\begin{figure}[ht]
  \centering

  \subcaptionbox
  {\label{fig:bichromatic_inclusion_rcht:1}Rectilinear convex hull inclusion.}
  [.25\linewidth][c]
  {\includegraphics[page=1,scale=0.9]{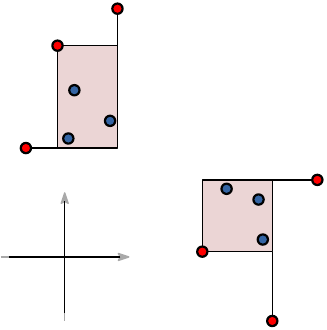}}
  \qquad
  \subcaptionbox{\label{fig:bichromatic_inclusion_rcht:2}Partial containment.}
  [.25\linewidth][c]
  {\includegraphics[page=2,scale=0.9]{bichromatic_inclusion_rcht}}
  \qquad
  \subcaptionbox{\label{fig:bichromatic_inclusion_rcht:3}The boundary of $\rch[R]$ separates $R$ from $B$.}
  [.26\linewidth][c]
  {\includegraphics[page=3,scale=0.9]{bichromatic_inclusion_rcht}}
  
  \caption
  {
    Relative position of $B$ and $\rch[R]$.
    The orientation of the coordinate axis is shown in the bottom left corner.
  }
  \label{fig:bichromatic_inclusion_rch}
\end{figure}

The main contribution of this paper is a time-optimal algorithm to compute a rectilinear convex hull with arbitrary orientation that is \emph{monochromatic}, i.e., that has no misclassified points.
We also provide similar results for generalizations of the rectilinear convex hull that stem from a variation of convexity known as \emph{restricted orientation convexity}~\cite{fink_2004,rawlins_thesis_1987} or \emph{$\o$-convexity}\footnote{In the literature, $\o$-convexity is also known as \emph{$D$-convexity}~\cite{schuirer_1991}, \emph{directional convexity}~\cite{franek_2008}, and \emph{set-theoretical $D$-convexity}~\cite{franek_2009}.}.
As we show, despite the separability problem seems harder in the context of $\o$-convexity than for standard convexity, under certain assumptions both cases can be solved within the same time and space complexities.

\subsection{Background and related work}

Restricted-orientation convexity in the Euclidean plane is a generalization of orthogonal convexity, and at the same time a restriction of standard convexity.
The \emph{orientation} of a line is the smallest of the two possible angles it makes with the $X^+$ positive semiaxis.
A \emph{set of orientations} $\o$ is a set of lines with different orientations passing through some fixed point.
A region of the plane is called \emph{$\o$-convex} if its intersection with any line parallel to a line of~$\o$ is either empty, a point, or a line segment.
Since this notion of convexity was defined in the early eighties, several results of topological and combinatorial flavors can be found in the literature, as well as computational problems that are usually adaptations of well-known problems related to standard convexity~\cite{fink_2004,martinez_2021}.

The \emph{$\o$-convex hull} of a finite point set is an $\o$-convex superset of such point set that generalizes both the standard and the rectilinear convex hull; refer to \cref{subsec:och} for a formal definition.
The $\o$-convex hull is relevant for research fields that require restricted-orientation enclosing shapes~\cite{daymude_2018}.
In the particular case where $\o$ is formed by two orthogonal lines, $\o$-convexity is known as \emph{orthogonal convexity}\footnote{In the literature, orthogonal convexity is also known as \emph{ortho-convexity}~\cite{rawlins_1988137} or \emph{x-y convexity}~\cite{nicholl_1983}.} and the $\o$-convex hull is known as the rectilinear convex hull.
The rectilinear convex hull has been extensively studied in the context of fields as diverse as polyhedra reconstruction~\cite{biedl_2011}, facility location~\cite{son_2014}, and geometric optimization~\cite{diaz_2011}; as well as in practical research fields such as pattern recognition~\cite{karmakar_2015},
shape analysis~\cite{biswas_2012}, and VLSI circuit layout
design~\cite{uchoa_2002}.

As far as we are aware, there are no previous results on the problem of separating bichromatic point sets by an $\o$-convex hull while the orientations of the lines of~$\o$ are changing.
Nevertheless, if the lines are fixed, then the problem can be trivially solved by combining the algorithm from Alegr\'ia et al.~\cite{alegria_2020} to compute the $\o$-convex hull of a finite set of $n$ points in $O(n \log n)$ time, and a straightforward extension of the so-called \emph{staircase structure} used to store the vertices of the rectilinear convex hull~\cite[Section 4.1.3]{preparata_1985}.
With this approach we obtain an $O(n \log n)$ time and $O(n)$ space algorithm to decide if there is a monochromatic $\o$-convex hull for any fixed orientations of the lines of $\o$.

The problem of separating a bichromatic point set using an $\o$-convex separator has already been studied for the particular case of orthogonal convexity.
In this case the problem consists in computing, if any, an orthogonally-convex
geometric separator for $R$ and $B$ among all possible orientations of the coordinate axes.
The most popular separator is the axis-aligned rectangle.
For $n = \vert R \vert + \vert B \vert$, an arbitrarily-oriented separating rectangle can be found in $O(n\log n)$ time and $O(n)$ space~\cite{van_kreveld_2009}.
Several variations have also been solved including separability by two disjoint rectangles~\cite{moslehi_2016}, bichromatic sets of imprecise points~\cite{sheikhi_2017}, maximizing the area of the separating rectangle~\cite{acharyya_2019,armaselu_2017}, and an extension where the separator is a box in three dimensions~\cite{hurtado_2005}.
Along with the axis-aligned rectangle, two more ortho-convex separators can be found in the literature.
In~\cite{sheikhi_2015} the authors use as separator an axis-aligned $L$-shaped region and solve the problem in $O(n^2)$ time.
In~\cite{pelaez_2011} the authors use as separator an alternating orthogonal polygonal chain, and also solve the problem in $O(n^2)$ time.

Our separability problem can also be considered as an instance of a general class of problems which consist in computing the orientations where an orientation-dependent geometric object satisfies some optimization criteria.
Our problem can then be stated as the problem of computing the orientations of the lines of $\o$ for which the $\o$-convex hull of $R$ has the minimum number of misclassified points.
If such a number is different from zero, then the given point sets cannot be separated by the particular $\o$-convex hull.
In this context, the $\o$-convex hull is called a \emph{weak separator} for $R$ and $B$.
The concept of weak separability was introduced by Houle~\cite{houle_1989,houle_1993}.
Separability results in this direction have been explored using $\o$-convex separators such as hyperplanes, strips, and rectangles~\cite{aronov_2012,cortes_2009,houle_1993,mangasarian_1994}.

Besides geometric separability, other similar types of problems can also be found in the literature.
Given a set $P$ of $n$ points in the plane, in~\cite{alegria_2020} the authors compute the angle by which the lines of $\o$ have to be simultaneously rotated around the origin for the $\o$-convex hull of $P$ to have minimum area.
A similar problem is solved in~\cite{alegria_2018}, where the authors compute the values of $\beta$ for which the $\ob$-convex hull of $P$ has maximum area, among other optimization criteria (refer to \cref{subsec:obh} for a formal definition of the $\ob$-convex hull).
More recently, in~\cite{bae_2019} the authors solved the problem of computing the set of empty squares with arbitrary orientations among a set of points.
From this result they derive an algorithm to compute the square annulus with arbitrary orientation of optimal width or area that encloses $P$, among other algorithmic results.

\subsection{Results}

In this paper we contribute with the following results:
\begin{itemize}
\item

  An optimal $O(n \log n)$ time and $O(n)$ space algorithm to compute a monochromatic rectilinear convex hull with arbitrary orientation, where $n = \vert R \vert + \vert B \vert$.
  
\item

  An algorithm to compute a monochromatic $\o$-convex hull with arbitrary orientation for a set $\o$ of $k \geq 2$ lines.
  In the counter-clockwise circular order of the lines of $\o$, let $\alpha_i$ be the angle required to clockwise rotate the $i$th line around the origin so it coincides with its successor.
  The algorithm runs in $O(\sfrac{1}{\Theta} \cdot N \log N)$ time and $O(\sfrac{1}{\Theta}\cdot N)$ space, where $\Theta = \min \{ \alpha_1,\ldots,\alpha_k \}$ and $N=\max\{ k, \vert R \vert + \vert B \vert \}$.
  
\item

  An optimal $O(n \log n)$ time and $O(n)$ space algorithm to compute the values of $\beta$ for which there is a monochromatic $\ob$-convex hull.
  
\item
  
  In all the cases, if there is no orientation of separability, the algorithms can be easily adapted to compute the hull that minimizes the number of misclassified points.

\end{itemize}

\subsection{Adopted conventions}

Throughout the rest of the paper, we denote with $R$ and $B$ two disjoint sets of red and blue points in the plane and denote $n = \vert R \vert + \vert B \vert$.
For the sake of simplicity, we assume that
the set $R \cup B$ contains no three points on a line.
Regarding the set of orientations, we assume for the sake of simplicity that all the lines of $\o$ have different orientations and pass through the origin.
We also assume that $\o$ contains a finite number of lines, and denote $k = \vert \o \vert$.
We remark that sets of orientations with an infinite number of lines have been considered in the literature~\cite{fink_2004,rawlins_thesis_1987}.
Finally, in our algorithms we adopt the real RAM model of computation~\cite{preparata_1985}, which is customary in computational geometry and allows us to perform standard arithmetic and trigonometric operations in constant time.

\subsection{Outline of the paper}

In \Cref{sec:rch} we solve the separability problem using a rectilinear convex hull with arbitrary orientation.
In \Cref{sec:generalizations} we solve the separability problem using an  $\ob$-convex hull, and an $\o$-convex hull with arbitrary orientation where the set $\o$ contains $k \geq 2$ lines.
Finally, we dedicate \Cref{sec:lower_bounds} to prove lower bounds.


\section{The rectilinear convex hull}\label{sec:rch}

In this section we solve the following problem.

\begin{problem}
  \label{problem:rch}
  Given a set of orientations $\o$ formed by $k=2$ orthogonal lines, compute the set of rotation angles for which the lines of $\o$ have to be simultaneously rotated around the origin in the counterclockwise direction, so the rectilinear convex hull of $R$ contains no points of $B$.
\end{problem}

We start with a formal definition of the rectilinear convex hull.
For the sake of completeness, we also briefly describe the properties of the rectilinear convex hull that are relevant to solve \cref{problem:rch}.
More details on these and other properties can be found in~\cite{fink_2004,ottmann_1984}.

Let $\rho_1$ and $\rho_2$ be two rays leaving a point $x\in\R$ such that, after rotating $\rho_1$ around $x$ by an angle of $\theta\in [0, 2\pi)$, we obtain $\rho_2$.
We refer to the two open regions in the set $\R \setminus (\rho_1 \, \cup \, \rho_2)$ as \emph{wedges}.
We say that both wedges have vertex $x$ and sizes $\theta$ and $2\pi -\theta$, respectively.
Throughout this section assume that the orientation set $\o$ is formed by two orthogonal lines.
A \emph{quadrant} is a wedge of size $\frac{\pi}{2}$ whose rays are parallel to the lines of $\o$.
Let $P$ denote a finite set of points in the plane.
We say a region of the plane is \emph{free of points of $P$}, or \emph{$P$-free} for short, if there are no points of $P$ in its interior.
The \emph{rectilinear convex hull} of $P$, denoted with $\rch$, is the set
\[
  \rch = \mathbb{R}^{2} \setminus \bigcup_{q\in\mathcal{Q}}q,
\]
where $\mathcal{Q}$ denotes the set of all $P$-free quadrants of the plane. See \Cref{fig:rcht}.

\begin{figure}[ht]
  \centering%

  \subcaptionbox
  {\label{fig:rcht:1}
    $\rch$ is formed by three connected components, one of which is a single point of $P$.
  }
  [.3\linewidth][c]{\includegraphics[page=1,scale=0.9]{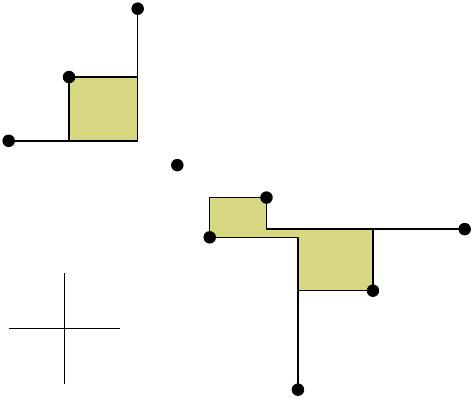}}%
  \qquad
  \subcaptionbox
  {\label{fig:rcht:2}
    $\rch$ is formed by four connected components.
  }
  [.3\linewidth][c]{\includegraphics[page=2,scale=0.9]{rcht}}%

  \caption
  {
    A finite point set $P$ and $\rch$ for two different rotation angles of the lines of $\o$.
  }
  \label{fig:rcht}
\end{figure}

Note that $\rch$ is not convex if at least one edge of the standard convex hull of $P$ is not parallel to a line of $\o$.
Moreover, $\rch$ may be disconnected.
Each connected component is either a single point of $P$, or a closed orthogonal polygon whose edges are parallel to a line of $\o$.
The rectilinear convex hull has also at most four ``degenerate edges'', which are orthogonal polygonal chains connecting either two extremal vertices, or a connected component to an extremal vertex.
Of special relevance is the property we call \emph{orientation dependency}: except for some particular cases, like rotating the orientations by $\frac{\pi}{2}$, the $\mathcal{RH}(P)$ at different orientations of the lines of $\o$ are non-congruent to each other.

Let $\o_\theta$ denote the set of lines obtained after simultaneously rotating the lines of $\o$ around the origin in the counter-clockwise direction by an angle of $\theta$.
We denote with $\rcht$ the rectilinear convex hull of $P$ computed with respect to $\o_\theta$.
We solve \cref{problem:rch} by describing an algorithm to compute the (possibly empty) set of angular intervals of~$\theta$ for which $\rcht[R]$ is $B$-free. Note that we are considering strict containment, so a blue point lying on the boundary of $\rcht[R]$ is not contained in $\rcht[R]$. Our algorithm runs in $O(n\log n)$ time and $O(n)$ space.
These are the same complexities required to compute the rectilinear convex hull of a set of $n$ points for a fixed orientation of the lines of $\o$~\cite{ottmann_1984}.

\subsection{Maximal wedges and maximal arcs}
\label{subsec:maximal_wedges_and_arcs}

Before describing our algorithm, we need some auxiliary results.
We start with the following proposition, which derives directly from the definition of the rectilinear convex hull.

\begin{proposition}\label{prop:rch_inclusion}
  A point $x \in \R$ is contained in $\rcht$ if, and only if, every quadrant with vertex on $x$ contains at least one point of $P$.
\end{proposition}

Let $w_x$ be a $P$-free wedge with vertex at a point $x\in\R$.
We say that $w_x$ is \emph{maximal}, if no other $P$-free wedge with vertex on $x$ intersects $w_x$.
Assume that $w_x$ is maximal.
Let $w_o$ be the wedge resulting from translating $w_x$ so that its vertex lies on the origin.
The \emph{maximal arc} of $x$ induced by $w_x$ is the circular arc that results from the intersection of $w_o$ and $\mathbb{S}^1$ (the unit circle centered at the origin).
Note that, since wedges (and hence, quadrants) are open regions, \cref{prop:rch_inclusion} excludes points on the boundary of $\rcht$, and the endpoints of a maximal arc do not belong to the maximal arc itself.
See \Cref{fig:maximal_arc}.

\begin{figure}[ht]
  \centering%

  \subcaptionbox
  {\label{fig:maximal_arc:1}
    A $P$-free maximal wedge $w_x$ with vertex on $x$.
    Note that $w_x$ has a point of $p$ lying on each of its bounding rays.
  }
  [.4\linewidth][c]{\includegraphics{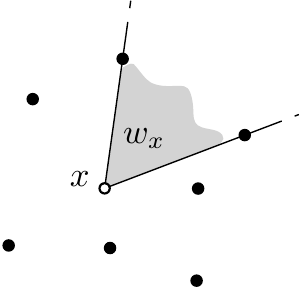}}%
  \qquad
  \subcaptionbox
  {
    \label{fig:maximal_arc:2}
    The wedge $w_o$ resulting from translating $w_x$ so that its vertex lies on the origin.
    The maximal arc of $x$ induced by $w_x$ is represented by the thick circular arc, and $\mathbb{S}^1$ by the dashed circle.
  }
  [.4\linewidth][c]{\includegraphics{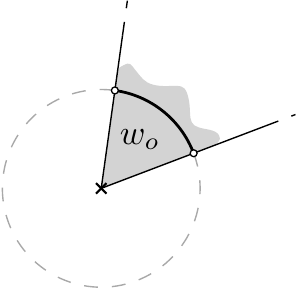}}%

  \caption{A maximal arc of a point $x\in\R$.}
  \label{fig:maximal_arc}
\end{figure}

A maximal arc is \emph{feasible} if it is induced by a maximal wedge with size at least $\frac{\pi}{2}$.
Hereafter, we consider $\o_\theta$ to be not only a set of two orthogonal lines, but also the set of four rays in which the orthogonal lines are split by the origin.

\begin{lemma}\label{lem:rch_inclusion_arcs}
  For any fixed value of $\theta$, a point $x \in \R$ is contained in $\rcht$ if, and only if, every feasible maximal arc of $x$ is intersected by a single ray of $\o_\theta$.
\end{lemma}

\begin{proof}
  We show that every quadrant with vertex on $x$ contains at least one point of $P$ if, and only if, every feasible maximal arc of $x$ is intersected by a single ray of $\o_\theta$.
  The lemma follows from this fact and \Cref{prop:rch_inclusion}.
  In the following, we assume without loss of generality that $\theta=0$ and $x$ lies on the   origin, so the lines of $\o_\theta$ coincide with the coordinate axes and every quadrant with vertex on $x$ is bounded by two coordinate semi-axes.
  
  \vspace{0.5em}%
  \noindent%
  $(\Longrightarrow)$ Using Proposition~\ref{prop:rch_inclusion}, assume that every quadrant with vertex on $x$ contains at least one point of $P$.
  We show that every feasible maximal arc of $x$ is intersected by a single ray of~$\o_\theta$.
  Let $w$ be a maximal wedge with vertex at $x$ and size at least $\frac{\pi}{2}$, and let $a$ be the feasible maximal arc of~$x$ induced by $w$.
  Since the size of $w$ is at least $\frac{\pi}{2}$, then $w$ contains at least one coordinate semi-axis.
  On the other hand, $w$ cannot contain two coordinate semi-axis, since otherwise $w$ would
  contain a $P$-free quadrant.
  This would be a contradiction, since we assumed that every quadrant with vertex on $x$ contains at least one point of $P$.
  Hence $w$ contains exactly one coordinate semi-axis, and thus, $a$ is intersected by a single ray of $\o_\theta$.
  
  \vspace{0.5em}%
  \noindent%
  $(\Longleftarrow)$ Assume that every feasible maximal arc of $x$ is intersected by a single ray of $\o_\theta$.
  We show that every quadrant with vertex on $x$ contains at least one point of $P$, which is enough by Proposition~\ref{prop:rch_inclusion}.
  For the sake of contradiction, suppose there is a $P$-free quadrant $q$ with vertex on $x$.
  Then there is a maximal wedge $w$ with vertex on $x$ that contains $q$.
  Since wedges are open regions, if the size of $w$ is equal to $\frac{\pi}{2}$ then $w$ contains no coordinate semi-axis, and thus, it induces a feasible maximal arc intersected by no ray of $\o_\theta$.
  On the other hand, if the size of $w$ is greater than $\frac{\pi}{2}$, then $w$ contains at least two coordinate semi-axes, and thus, it induces a feasible maximal arc of $x$ intersected by at least two rays of $\o_\theta$.
  Either case is a contradiction, since we assumed that every feasible maximal arc of $x$ is intersected by a single ray of $\o_\theta$.
\end{proof}

An illustration of \cref{lem:rch_inclusion_arcs} is shown in \cref{fig:inclusion_arcs}.
In the figure we can see a set $P$ of four points, $\rcht$, a point $x\in\R$, the feasible maximal arcs of $x$, and the lines of $\o_\theta$.
The maximal arcs are drawn with thick circular arcs.
Instead of drawing the arcs on a single circle representing $\mathbb{S}^1$, we draw them separately on concentric circles for the sake of clarity.
In \cref{fig:inclusion_arcs:1} the point $x$ is not contained in $\rcht$; hence, there is at least one feasible maximal arc of $x$ that is not intersected by a single ray of $\o_\theta$.
Note that the maximal wedge $w$ induces a feasible maximal arc $a$ that is intersected by two rays of $\o_\theta$.
In \cref{fig:inclusion_arcs:2} the point $x$ is contained in $\rcht$; hence, all the feasible maximal arcs of $x$ are intersected by a single ray of $\o_\theta$.

\begin{figure}[ht]
  \centering

  \subcaptionbox
  {\label{fig:inclusion_arcs:1}
    If either zero or at least two rays of $\o_\theta$ intersect a feasible maximal arc of $x$, then $x$ is not contained in $\rcht$.
  }
  [0.4\linewidth][c]{\includegraphics[page=2,scale=1]{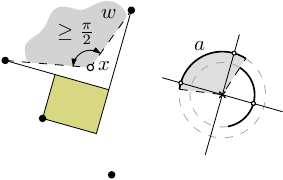}}
  \qquad
  \subcaptionbox
  {\label{fig:inclusion_arcs:2}
    If all the feasible maximal arcs of $x$ are intersected by a single ray of $\o_\theta$, then $x$ is contained in $\rcht$.
  }
  [0.4\linewidth][c]{\includegraphics[page=3,scale=1]{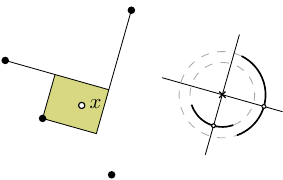}}

  \caption{Containment of a point $x\in\R$ in $\rcht$.}
  \label{fig:inclusion_arcs}
\end{figure}

The adaptation of \Cref{lem:rch_inclusion_arcs} to a bichromatic setting is straightforward.
A \emph{blue maximal wedge} is an $R$-free maximal wedge with vertex on a blue point.
A \emph{blue maximal arc} is a maximal arc induced by a blue maximal wedge.
A blue maximal arc is \emph{feasible} if it is induced by a blue maximal wedge with size
at least $\frac{\pi}{2}$.

\begin{lemma}\label{lem:rch_bichromatic_inclusion_arcs}
  A blue point $b \in B$ is contained in $\rcht[R]$ if, and only if, every blue maximal arc of $b$ that is feasible is intersected by a single ray of $\o_\theta$.
\end{lemma}

Let $\widehat{D}$ be a direction in $\mathbb{S}^1$ and let $w$ be a $P$-free maximal wedge with vertex on a point $x \in \R$.
We say that $w$ is \emph{constrained to} $\widehat{D}$ if $w$ contains the ray leaving $x$ with direction~$\widehat{D}$.
We compute the set of blue maximal arcs that are feasible by means of a procedure to compute the set of blue maximal wedges constrained to a given direction.
This procedure is an adaptation for bichromatic point sets of the \emph{restricted unoriented maximum} approach from Avis et~al.~\cite{avis_1998}.
Given a set $P$ of $n$ points in the plane and an angle $\Theta \geq \sfrac{\pi}{2}$, the authors compute, in $O(n \log n)$ time and $O(n)$ space, the set of $P$-free wedges with size at least $\Theta$ and vertex on a point of $P$.

The adapted procedure is as follows.
Let $\widehat{D}$ denote the direction given as input.
Without loss of generality, assume that $\widehat{D}$ is equal to the $Y^+$ semiaxis.
We first sort the points of the set $R \cup B$ in a direction orthogonal to $\widehat{D}$ (along the $X$ axis in our assumption).
We then perform two sweeps on the sorted set of points.
In the first sweep we traverse the points from left to right.
A red point is processed using an on-line algorithm to construct the convex hull of the red visited points, one point at a time.
To process a blue point $b$, we compute the $R$-free wedge with vertex on $b$ that is bounded by a ray leaving $b$ with direction $\widehat{D}$, and the tangent from $b$ to the red convex hull.
In the second sweep we traverse the sorted set of points from right to left and process points in a symmetric way.
Let $w_l(b)$ and $w_r(b)$ denote, respectively, the $R$-free wedges obtained after processing a blue point $b$ in the sweeps from left-to-right and from right-to-left.
After performing both sweeps, we report $w_l(b) \cup w_r(b)$ as a blue maximal wedge constrained to $\widehat{D}$, for all $b \in B$.
See \Cref{fig:restricted_maximal_wedge}.

\begin{figure}[ht]
  \centering
  \includegraphics{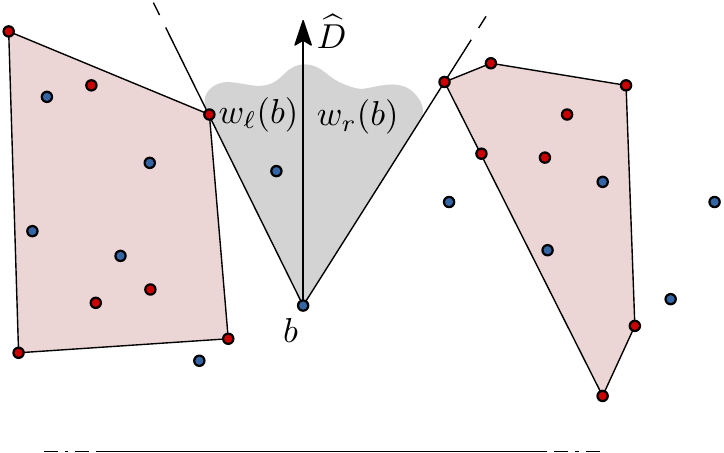}
  \caption{The $R$-free maximal wedge with vertex on $b$ constrained to $\widehat{D}$.}
  \label{fig:restricted_maximal_wedge}
\end{figure}

In the procedure described above, we first sort the points in the direction orthogonal to~$\widehat{D}$ in $O(n \log n)$ time.
Using standard techniques~\cite{preparata_1985}, during each sweep we process a point in $O(\log \vert R \vert) = O(\log n)$ time: If a red point, we are updating the convex hull of a point set by inserting a new point.
If a blue point, we are computing the tangent from a point to a convex polygon described by the sorted list of its vertices.
Since each blue point is the vertex of a single $R$-free maximal wedge constrained to~$\widehat{D}$, the whole procedure takes $O(n \log n)$ time and $O(n)$ space.
We obtain the following lemma.

\begin{lemma}\label{lem:restricted_maximal_wedges}
  Given a direction $\widehat{D}$ in $\mathbb{S}^1$ and two disjoint sets $R$ and $B$ of red and blue points in the plane, the set of $\vert B \vert$ blue maximal wedges constrained to $\widehat{D}$ can be computed in $O(n \log n)$ time and $O(n)$ space, where $n = \vert R \vert + \vert B \vert$.
\end{lemma}

And we obtain the following result.

\begin{lemma}\label{lem:feasible_maximal_arcs}
  There are $O(n)$ blue maximal arcs that are feasible. The set of
  blue maximal arcs that are feasible can be computed in $O(n \log n)$
  time and $O(n)$ space, where $n = \vert R \vert + \vert B \vert$.
\end{lemma}

\begin{proof}
  A maximal arc is induced by a blue maximal wedge with size at least $\frac{\pi}{2}$.
  Since a blue point is the vertex of at most four of such wedges, then each blue point has at most four blue maximal arcs that are feasible.
  Hence, there are $O(\vert B \vert) = O(n)$ arcs.

  We compute the set of blue maximal arcs that are feasible as follows.
  Note that a maximal wedge with size at least $\frac{\pi}{2}$ is constrained to one of the $X^+$, $X^-$, $Y^+$, or $Y^-$ coordinate semiaxis.
  In $O(n \log n)$ time and $O(n)$ space, we compute the set of blue maximal wedges constrained to each coordinate semiaxis, by means of the algorithm used to prove \cref{lem:restricted_maximal_wedges}.
  Then, we traverse the resulting set of $O(\vert B \vert) = O(n)$ blue maximal wedges, and keep those with size at least $\frac{\pi}{2}$.
  Finally, we transform each maximal wedge into a maximal arc in $O(1)$ time per wedge.
  Since the most expensive step is the computation of the set of blue maximal wedges, the whole procedure takes $O(n \log n)$ time and $O(n)$ space.
\end{proof}

\subsection{The algorithm}
\label{sec:rch_algorithm}

We are now ready to describe the algorithm to compute the set of angular intervals of $\theta\in(0,2\pi]$ for which $\rcht[R]$ is $B$-free.
Our strategy is to perform an angular sweep on the set of blue maximal arcs that are feasible, while we maintain the set $B_\theta$ of blue points in the interior of $\rcht[R]$.
To perform the angular sweep we increment $\theta$ from $0$ to $\pi/2$, so the four rays of $\o_\theta$ sweep all the directions of $\mathbb{S}^1$.
By \cref{lem:rch_bichromatic_inclusion_arcs}, for a particular value of $\theta$ during the sweep process, a blue point $b$ is contained in $B_\theta$ if all the blue maximal arcs of $b$ that are feasible are intersected by a single ray of $\o_\theta$.
Hence, $B_\theta$ only changes at the values of $\theta$ where a ray of $\o_\theta$ passes over an endpoint of a maximal arc.
We call these rotation angles \emph{intersection events}.
By means of a set of $\vert B \vert$ auxiliary variables, we update $B_\theta$ at each intersection event in constant time.
The algorithm is described in detail next.

\subparagraph{Step 1.  Computing the set of feasible maximal arcs.}

The first step of the algorithm is to compute the set $\mathcal{A}$ of $O(\vert B \vert) = O(n)$ blue maximal arcs that are feasible.
We compute this set by means of the procedure used to prove \cref{lem:feasible_maximal_arcs}.
Hence, this step takes $O(n \log n)$ time and $O(n)$ space.

\subparagraph*{Step 2. Computing the list of intersection events.}

The second step is to transform the set of blue maximal arcs that are feasible into a sorted circular list $\mathcal{L}$ of intersection events.
Since intersection events are given by the endpoints of maximal arcs, each maximal arc is transformed into two intersection events, hence there are $O(\vert B \vert) = O(n)$ intersection events.
Let $a$ be a blue maximal arc that is feasible, and let $p$ and $q$ be the endpoints of $a$.
We transform $a$ into a pair of intersection events by computing, in $O(1)$ time, the directions in $\mathbb{S}^1$ of the rays leaving the origin that pass through $p$ and $q$, see \cref{fig:arc_to_events}.
We can thus transform the set of blue maximal arcs that are feasible into the set of intersection events in $O(n)$ time.
We store the set of $O(n)$ intersection events in $\mathcal{L}$, sorted as the endpoints of the maximal arcs appear while traversing $\mathbb{S}^1$ in the counter-clockwise direction.
Since the most expensive task is to sort the set of intersection events, this step takes $O(n \log n)$ time and $O(n)$ space.

\begin{figure}[ht]
  \centering

  \subcaptionbox
  {\label{fig:arc_to_events:1}
    A blue maximal wedge $w$ with size at least $\frac{\pi}{2}$.
  }
  [0.3\linewidth][c]{\includegraphics[page=1]{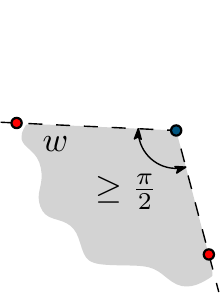}}
  \qquad
  \subcaptionbox
  {\label{fig:arc_to_events:2}
    The feasible maximal arc $a$ induced by $w$.
  }
  [0.3\linewidth][c]{\includegraphics[page=2]{intersection_events}}
  
  \caption
  {
    The endpoints $p$ and $q$ of $a$ can be transformed in $O(1)$ time into two intersection events $\theta_p$ and $\theta_q$, respectively.
  }
  \label{fig:arc_to_events}
\end{figure}

\subparagraph*{Step 3. Performing the angular sweep.}

The final step is to perform an angular sweep on the set of blue maximal arcs that are feasible.
Let $b_1,\ldots,b_{\vert B \vert}$ be the set of blue points labeled with no particular order.
Let $N_i$, $0 \leq N_i \leq 4$, denote the number of blue maximal arcs of the point $b_i$ that are feasible, and let $n_i$, $0 \leq n_i \leq N_i$, denote the number of blue maximal arcs of $b_i$ that are intersected by a single ray of $\o_\theta$.
We use an array of $\vert B \vert$ Boolean flags to represent if a blue point belongs to $B_\theta$, so
the status of a blue point can be changed in $O(1)$ time.
Following the condition from \cref{lem:rch_bichromatic_inclusion_arcs}, we set the $i$-th flag of the array to \texttt{True} if $n_i = N_i$ ($b_i$ belongs to $B_\theta$), and to \texttt{False} if $n_i < N_i$ ($b_i$ does not belong to $B_\theta$).

To process intersection events during the angular sweep we use the following auxiliary structures.
For each blue maximal arc $a$ that is feasible, we define a variable $\rho(a)$ that contains the number of rays of $\o_\theta$ currently intersecting $a$.
We use a min-priority queue $\mathcal{Q}$ to predict the next intersection event, among the events induced by all the blue maximal arcs.
Let $r_1,\ldots,r_4$ be the rays of $\o_\theta$ sorted in counter-clockwise circular order around the origin, and let $\theta_i$ be the smallest rotation angle for which $r_i$ passes over an endpoint of a maximal arc.
The queue contains the angles $\theta_1,\ldots,\theta_4$ that are less than $\frac{\pi}{2}$.
The next intersection event is thus given by the minimum element in $\mathcal{Q}$.
Since $\mathcal{Q}$ contains at most four elements, both update and query operations on $\mathcal{Q}$ take $O(1)$ time.
See \cref{fig:angular_sweep}.

\begin{figure}[ht]
  \centering

  \subcaptionbox
  {\label{fig:angular_sweep:1}
    The point $b$ has two blue feasible maximal arcs.
    These arcs induce the sorted sequence of intersection events $\{ \alpha, \beta, \gamma \}$.
  }
  [0.46\linewidth][c]{\includegraphics[page=1,scale=0.95]{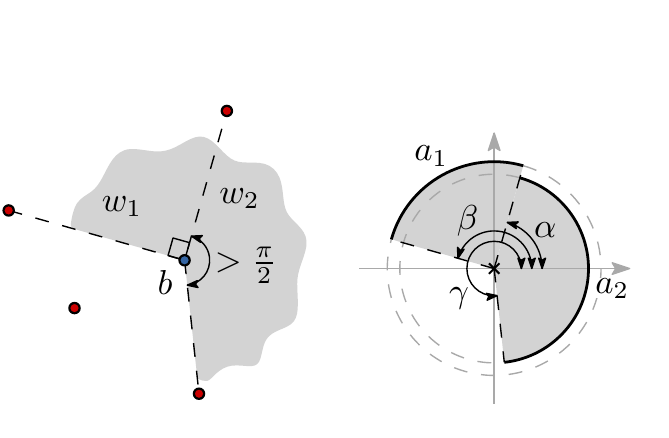}}
  \qquad
  \subcaptionbox
  {\label{fig:angular_sweep:2}
    For $\theta=0$ the point $b$ is contained in $\rcht[R]$, since all its feasible maximal arcs are intersected by a single ray of $\o_\theta$.
    Since $\theta_3 \geq \frac{\pi}{2}$, only the angles $\theta_1$, $\theta_2$, and $\theta_4$ are stored in $\mathcal{Q}$.
    The next intersection event is given by $\theta_4$.
  }
  [0.46\linewidth][c]{\includegraphics[page=2,scale=0.95]{rch_sweep}}
  \\[1em]
  \subcaptionbox
  {\label{fig:angular_sweep:3}
    A value of $\theta$ for which $b$ is not contained in $\rcht[R]$.
    Since $\theta_4 \geq \frac{\pi}{2}$, only the angles $\theta_1$, $\theta_2$, and $\theta_3$ are stored in $\mathcal{Q}$.
    The next intersection event is given by $\theta_1$.
  }
  [0.46\linewidth][c]{\includegraphics[page=3,scale=0.95]{rch_sweep}}

  \caption{Illustration of the angular sweep. For the sake of clarity, we consider a set $B$ with a single blue point $b$.}
  \label{fig:angular_sweep}
\end{figure}

We now describe how to perform the angular sweep.
First, we initialize the auxiliary data structures described above at an initial value of
$\theta$, say $\theta=0$.
Consider the four rays of $\o_\theta$ sorted in counter-clockwise circular order around the origin.
We first merge, in $O(n)$ time, the angles in $\mathcal{L}$ with the orientation angles given by the sorted set of rays of $\o_\theta$.
After merging, we can say which blue maximal arcs are intersected by each ray of $\o_\theta$, as well as the smallest rotation angle for which each ray of $\o_\theta$ passes over the endpoint of a maximal arc.
Using this information, in $O(n)$ time we compute the values of the variables $n_i$, $1 \leq i \leq \vert B \vert$, and $\rho(a)$ for all $a \in \mathcal{A}$, and initialize the set of $\vert B \vert$ Boolean flags we use to represent $B_\theta$.
We finally initialize $\mathcal{Q}$ in $O(1)$ time.
Hence, the whole initialization step takes $O(n)$ time.

We perform the angular sweep by incrementing $\theta$ from $0$ to $\sfrac{\pi}{2}$.
The next intersection event is obtained by extracting the minimum angle from $\mathcal{Q}$.
Consider an intersection event $\theta$ for which a ray $r\in\o_\theta$ is passing over the endpoint of a maximal arc $a$ of a blue point $b_i$.
We process the event as follows:
\begin{itemize}
\item
  If $r$ starts intersecting $a$, then we increase $\rho(a)$ by one.
  If, instead, $r$ stops intersecting $a$, then we decrease $\rho(a)$ by one.
\item
  If $\rho(a)$ was changed, then we update $n_i$.
  If $\rho(a)$ is equal to one, then we increase $n_i$ by one.
  If $\rho(a)$ is instead different from one, then we decrease $n_i$ by one.
\item
  If $n_i$ was changed, then we update the Boolean flags that represent $B_\theta$. If $n_i = N_i$, then we set the $i$-th flag of $B_\theta$ to \texttt{True}.
  If instead $n_i < N_i$, then we set the $i$-th flag to \texttt{False}.
\item
  Finally, we obtain from $\mathcal{L}$ the successor $\theta^{\prime}$ of $\theta$, and insert the angle $\theta^{\prime} - \theta$ into $\mathcal{Q}$.
\end{itemize}

\begin{lemma}\label{lemma:rch}
  The set $B_\theta$ can be computed and maintained while $\theta$ is increased from $0$ to $\sfrac{\pi}{2}$ in $O(n \log n)$ time and $O(n)$ space, where $n = \vert R \vert + \vert B \vert$.
\end{lemma}

\begin{proof}
  Steps 1 and 2 take $O(n \log n)$ time and $O(n)$ space.
  Since we have $O(n)$ intersection events and each event is processed in $O(1)$ time, the sweep process of Step 3 takes $O(n)$ time and $O(n)$ space.
  The lemma follows.
\end{proof}

By keeping track of the changes of $B_\theta$ we can construct the angular intervals for which all the flags of $B_\theta$ are \texttt{False}.
Hence, from \cref{lemma:rch} we obtain the main result of this section.

\begin{theorem}\label{thm:rch_separability}
  Given two disjoint sets $R$ and $B$ of points in the plane, the (possibly empty) set of angular intervals of $\theta \in [0,2\pi)$ for which $\rcht[R]$ is $B$-free (i.e., $\rcht[R]$ is a separator of $R$ and $B$) can be computed in $O(n\log n)$ time and $O(n)$ space, where $n = \vert R \vert + \vert B \vert$.
\end{theorem}

The algorithm we described to prove \Cref{thm:rch_separability} is time-optimal.
A proof of the $\Omega(n\log n)$-time lower bound is presented in \cref{sec:lower_bounds}.
There are a couple of additional facts worth mentioning.
First, the algorithm only computes a set of angular intervals.
To actually compute a monochromatic rectilinear convex hull we need to first choose an angle in one of these intervals, and then spend additional $O(n\log n)$ time~\cite{alegria_2020,diaz_2011,ottmann_1984}.
Second, the reported angular intervals are maximal in the sense that no two of them intersect each other.
The intervals are also open since they are bounded by intersection events and, at such events, a blue point lies on the boundary of some $R$-free quadrant.
Hence, the point lies on the boundary of the rectilinear convex hull of $R$.
Finally, since there is at most one change in~$B_\theta$ per intersection event and there are $O(n)$ intersection events, then there are $O(n)$ angular intervals of $\theta$ where $\rcht[R]$ is $B$-free.
A matching lower bound is achieved by the point set we describe next.

\subsection{Lower bound for the number of intervals of separability}
\label{subsec:rch_separability_lower_bound}

In this subsection we describe a bichromatic point set with $\Omega(n)$ angular intervals of $\theta$ for which $\rcht[R]$ is $B$-free.
The first ingredient of the construction is the fact that $\rcht$ may be disconnected.
As previously mentioned, a connected component is either a single point of $P$, an orthogonal polygonal chain, or a closed orthogonal polygon.
The polygonal chain connects two extremal points of $P$ and contains exactly two segments.
The orthogonal polygon may have at most two ``degenerate edges'' in each direction, which are horizontal or vertical segments connecting its vertices with extremal points of $P$.
The segments of the polygonal chains and the edges of the orthogonal polygons are called the \emph{edges} of $\rcht$.
Each edge is contained in a ray of some $P$-free quadrant.
Such a $P$-free quadrant is said to be \emph{stabbing} $P$.
Note that each of the rays of a $P$-free quadrant stabbing $P$ contains an edge of $\mathcal{RH}_{\theta}(P)$.
See \Cref{fig:rch_disconnected:1}.

Let $r_1,\ldots,r_4$ be the rays of $\o_\theta$ labeled in counter-clockwise circular order around the origin.
For the sake of simplicity, in the following we assume an index $i \in \{ 1,\ldots,4\}$ is such that $i + 4 := i$.
Let $Q_i$ denote the quadrant bounded by $r_i$ and $r_{i+1}$.
A \emph{$Q_i$-quadrant} is a translation of $Q_i$.
We say that a $Q_i$-quadrant and a $Q_{i+2}$-quadrant are \emph{opposite} to each other.
The following lemma states the conditions in which $\rcht$ is disconnected.
See \Cref{fig:rch_disconnected}.

\begin{lemma}[Alegr\'{i}a et al.~\cite{alegria_2020}, Lemma~1]
  \label{lemma:opposite_quadrants}
  Let $i$ and $j$, $i \neq j$, be two indices in $\{1,\ldots,4\}$.
  Let $q_i$ be a $Q_i$-quadrant and $q_j$ be a $Q_j$-quadrant.
  If both $q_i$ and $q_j$ are stabbing $P$ and $q_i \, \cap \, q_j \, \cap \, \ch \neq \emptyset$, then the following statements hold true:
  \begin{enumerate}[(a)]
  \item The quadrants $q_i$ and $q_j$ are opposite to each other, that is $\vert j - i \vert = 2$.
  \item For all $k\in\{1,\ldots,4\}$, $k\neq i$, $k\neq j$, every $Q_k$-quadrant $q_k$ and every $Q_{k+2}$-quadrant $q_{k+2}$ are such that $q_k \, \cap \, q_{k+2} \, \cap \, \ch = \emptyset$.
  \item $\rcht$ is disconnected.
  \end{enumerate}
\end{lemma}

\begin{figure}[ht]
  \centering
  \subcaptionbox
  {\label{fig:rch_disconnected:1}
    A stabbing $Q_2$-quadrant and a stabbing $Q_4$-quadrant have a non empty intersection, hence they disconnect $\rcht$.
    The boundary of~$\ch$ is shown in blue.
    The quadrants are shown in light gray, their bounding rays in dashed lines, and their intersection in dark gray.
  }
  [.45\linewidth][c]{\includegraphics[page=1]{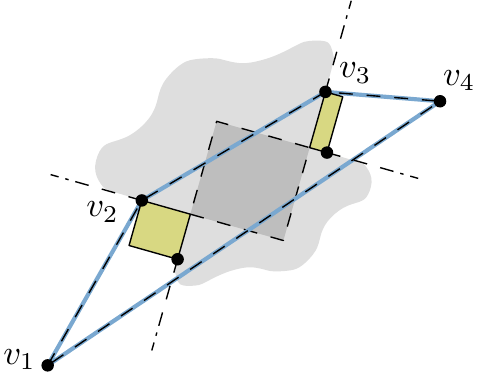}}
  \qquad
  \subcaptionbox
  {\label{fig:rch_disconnected:2}
    The set $\o_\theta$ and the directions of the edges of $\ch$; where $d_i$ is the direction of the edge $\overline{v_i v_{i+1}}$.
    All the directions lie in either $Q_1$ or~$Q_3$, hence $\rcht$ is disconnected because of the intersection of $Q_2$-quadrants and $Q_4$-quadrants.
  }
  [.45\linewidth][c]{\includegraphics[page=2]{rch_opposite_quadrants}}
  
  \caption
  {
    A set $P$ of six points for which $\rcht$ is formed by four connected components (two of which are points of $P$), and an orientation set $\o_\theta$, both for some value of $\theta$.
    Figure \subref{fig:rch_disconnected:1} illustrates \Cref{lemma:opposite_quadrants}, and Figure \subref{fig:rch_disconnected:2} illustrates \Cref{lemma:stabbing_quadrants}.
  }
  \label{fig:rch_disconnected}
\end{figure}

It is known that $\rcht$ is contained in $\ch$ regardless of the value of $\theta$~\cite[Theorem~4.7]{preparata_1985}.
Hence, a quadrant stabbing $P$ is necessarily intersecting $\ch$.
Let $u$ and $v$ be two vertices of $\ch$ such that $u$ precedes $v$ in the clockwise circular order of the vertices of $\ch$.
Let $r$ be the ray leaving $u$ passing through $v$.
The \emph{direction} of the edge $\overline{u v}$ is the translation of $r$ so that $u$ lies on the origin.
The following lemma is used in our construction to identify which of the four families of $Q_i$-quadrants can stab $P$.
Refer again to \Cref{fig:rch_disconnected}.

\begin{lemma}[Alegr\'{i}a et al.~\cite{alegria_2020}, Observation~2]
  \label{lemma:stabbing_quadrants}
  If a $Q_i$-quadrant is stabbing $P$, then there is at least one edge of $\ch$ whose direction is contained in $Q_i$.
\end{lemma}

We are now ready to describe the construction.
The convex hull of $R$ is a rhombus whose diagonals are parallel to the coordinate axes.
Let $v_1,\ldots,v_4$ be the vertices of $\ch[R]$ labeled in clockwise circular order starting at the left-most vertex.
The vertices $v_2$ and $v_4$ lie outside the circle $C$ that has the line segment $\overline{v_1 v_3}$ as diameter; thus, the interior angles of the rhombus at $v_2$ and $v_4$ are smaller than $\frac{\pi}{2}$, as well as the orientation $\alpha$ of the line through $v_3$ and $v_4$; see \Cref{fig:rch_separability_red_points:1}.

\begin{figure}[ht]
  \centering

  \subcaptionbox
  {\label{fig:rch_separability_red_points:1}The convex hull is a rhombus.}
  [.25\linewidth][c]{\includegraphics[page=1]{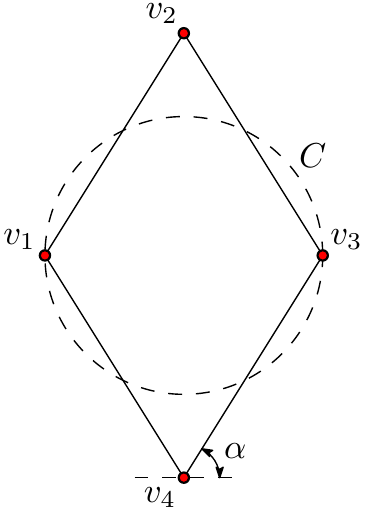}}%
  \hfill
  \subcaptionbox
  {\label{fig:rch_separability_red_points:2}
    The rectilinear convex hull has three connected components: the points $v_2$ and $v_4$, and a rectangle inscribed in the circle $C$ whose sides are parallel to the lines of $\o_\theta$.
  } [.45\linewidth][c]{\includegraphics[page=4,scale=0.8]{separability_rch_construction}\includegraphics[page=2]{separability_rch_construction}}%
  \hfill
  \subcaptionbox
  {\label{fig:rch_separability_red_points:3}
    The remaining points lie in a second rhombus contained in $C$.
  }
  [.25\linewidth][c]{\includegraphics[page=3]{separability_rch_construction}}

  \caption{A bichromatic point set with $\Omega(n)$ angular intervals of separability: The set of red points.}
  \label{fig:rch_separability_red_points}
\end{figure}

Let $d_i$ be the direction of the edge $\overline{v_i v_{i+1}}$.
Let $l_\alpha = d_1 \cup d_3$ and $l_{\pi-\alpha} = d_2 \cup d_4$ be respectively, the lines through the origin with orientations $\alpha$ and $\pi-\alpha$ formed by the orientations of the edges of $\ch[R]$; see \Cref{fig:rch_separability_red_points:2}, left.
Let $\ell_i$ denote the line of $\o_\theta$ that contains the rays $r_i$ and $r_{i+2}$.
While incrementing $\theta$ from $0$ to $\frac{\pi}{2}$, the lines of $\o_\theta$ counter-clockwise rotate around the origin while the lines $l_\alpha$ and $l_{\pi-\alpha}$ remain fixed.
The rotation angles that are relevant for the lower bound are those in the interval $\varphi=[\frac{\pi}{2}-\alpha, \alpha]$.
At the angle $\theta=\frac{\pi}{2}-\alpha$ the line $\ell_2$ coincides with $l_{\pi-\alpha}$.
At the angle $\theta=\alpha$ the line $\ell_1$ coincides with $l_\alpha$.
For any other rotation angle in $\varphi$, the directions $d_1$ and $d_4$ lie in $Q_1$, whereas $d_2$ and $d_3$ lie in $Q_3$.
Hence, by \Cref{lemma:stabbing_quadrants} the set $R$ is stabbed only by $Q_2$- and $Q_4$-quadrants for all $\theta\in\varphi$.
Note that $R$ is stabbed on all the edges of the rhombus, and the vertices of the stabbing quadrants lie on semicircles in the interior of the rhombus whose diameters are the edges of the rhombus.
Therefore, every point in the dashed regions lies in the intersection of a stabbing $Q_2$-quadrant and a stabbing $Q_4$-quadrant.
Since these quadrants are opposite to each other, we have by \Cref{lemma:opposite_quadrants} that $\rcht[\{v_1,\ldots,v_4\}]$ is disconnected for all $\theta\in\varphi$.
As shown in \Cref{fig:rch_separability_red_points:2}, right, $\rcht[\{v_1,\ldots,v_4\}]$ is actually formed by three connected components: the point~$v_2$, the point $v_4$, and a rectangle inscribed in $C$ whose sides are parallel to the lines of $\o_\theta$.
By intersecting all such rectangles for all the rotation angles in $\varphi$, we obtain the rhombus highlighted in \Cref{fig:rch_separability_red_points:3}.
Note that any red point lying in this region is contained in $\rcht[R]$ for all $\theta\in\varphi$.
Hence, we may add as many red points as desired without affecting the construction.

The set of blue points is shown in \Cref{fig:rch_separability_blue_points}.
Let $d(p,q)$ denote the Euclidean distance between two given points $p$ and $q$.
The blue points lie in the interior of $\ch[R]$, on a circle with center on the middle point of the segment $\overline{v_1 v_3}$, and radius $\sfrac{d(v_1,v_3)}{2} - \varepsilon$ for $0 < \varepsilon < \sfrac{d(v_1,v_3)}{2}$.
The points are spread so that, at every $\theta\in\varphi$, at most a single blue point is contained in $\rcht[R]$.
As shown in \Cref{fig:rch_separability_intervals:1} (see the figures from left to right), while incrementing $\theta$ from $\frac{\pi}{2}-\alpha$ to $\alpha$, two of the vertices of the rectangle inscribed in $C$ remain anchored at the red points, while the other two traverse the red circular arcs in the counter-clockwise direction.
Hence $\rcht[B]$ captures one blue point at a time, generating $\Omega(n)$ disjoint angular intervals of separability.

\begin{figure}[ht]
  \centering

  \subcaptionbox
  {\label{fig:rch_separability_blue_points}
    The points lie inside $\ch[R]$, on a circle with center on the middle point of $\overline{v_1 v_3}$ and radius $\sfrac{d(v_1,v_3)}{2} - \varepsilon$.
  }
  [.25\linewidth][c]{\includegraphics[page=1]{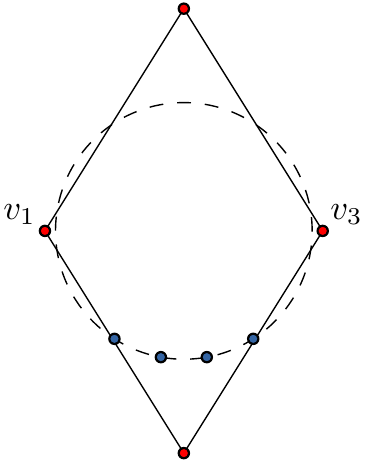}}%
  ~~
  \subcaptionbox
  {\label{fig:rch_separability_intervals:1}
     While incrementing $\theta$ from $\frac{\pi}{2}-\alpha$ to $\alpha$, $\rcht[R]$ captures one blue point at a time.
    The dashed rectangles are the intersections between the opposite quadrants that disconnect $\rcht[R]$.
  }
  [.758\linewidth][c]{%
    \includegraphics[page=3]{separability_rch_points}%
    \hspace{-1.1em}
    \includegraphics[page=4]{separability_rch_points}%
    \hspace{-1.1em}
    \includegraphics[page=5]{separability_rch_points}%
  }

  \caption{A bichromatic point set with $\Omega(n)$ angular intervals of separability: The set of blue points.}
  \label{fig:rch_separability_intervals}
\end{figure}

\subsection{Inclusion detection}
\label{subsec:rch_inclusion}

An elementary property of the standard convex hull is the following: $\ch[B]$ is in the interior of $\ch[R]$ if all the blue points are in the interior of $\ch[R]$.
This property translates to the rectilinear convex hull, regardless of the slopes of the lines of $\o_\theta$, and the connected components of both $\rcht[B]$ and $\rcht[R]$.

\begin{lemma}\label{lem:rch_containment}
If all the points of $B$ are in the interior of $\rcht[R]$, then $\rcht[B]$ is in the interior of $\rcht[R]$.
\end{lemma}

\begin{proof}
Suppose a fixed value of $\theta$ and that all the points of $B$ are in the interior of $\rcht[R]$. Let $x$ be a point in the plane in the interior of $\rcht[B]$. By \Cref{prop:rch_inclusion}, every $Q_i$-quadrant with vertex at $x$ contains at least one blue point. Let $w_x$ be a $Q_i$-quadrant with vertex at $x$, and $b$ denote one of the blue points contained in $w_x$. Let $w_b$ be the $Q_i$-quadrant resulting from translating $w_x$ so its vertex lies on $b$. Since we assumed all the blue points being in the interior of $\rcht[R]$, then $b$ is in the interior of $\rcht[R]$ and by \Cref{prop:rch_inclusion}, $w_b$ contains at least one red point $r$. Note that $w_x$ contains $r$ since $w_b\subset w_x$. Thus, every $Q_i$-quadrant with vertex at a point in the interior of $\rcht[B]$ contains at least one red point.
\end{proof}

By \Cref{lem:rch_containment}, if there is a value of $\theta$ for which all the Boolean flags of the set that encodes~$B_\theta$ are \texttt{True}, then $\rcht[B]$ is contained in $\rcht[R]$.
We obtain the following theorem as a consequence of \Cref{lemma:rch}.

\begin{theorem}\label{thm:rch_containment}
  Given two disjoint sets $R$ and $B$ of points in the plane, the (possibly empty) set of angular intervals of $\theta \in [0,2\pi)$ for which $\rcht[B]$ is contained in $\rcht[R]$ can be computed in $O(n \log n)$ time and $O(n)$ space, where $n = \vert R \vert + \vert B \vert$.
\end{theorem}

It is not hard to see that, as in the separability problem, there are $O(n)$ angular intervals of containment. We next adapt the bichromatic point set from Subsection~\ref{subsec:rch_separability_lower_bound} to obtain a matching lower bound.
Consider a rhombus whose diagonals are parallel to the coordinate axes, and a circle $C$ whose diameter is the diagonal of the rhombus that is parallel to the $X$ axis.
The convex hull of the set $R$ is now formed by the five points shown in \Cref{fig:rch_containment_red_points:1}.
The points $v_1$, $v_2$, and $v_5$ lie on vertices of the rhombus, and the points $v_3$ and $v_4$ on intersection points between the rhombus and the circle $C$.

\begin{figure}[ht]
  \centering

  \subcaptionbox
  {\label{fig:rch_containment_red_points:1}
    The points lie on the boundary of a rhombus.
  }
  [.25\linewidth][c]
  {\includegraphics{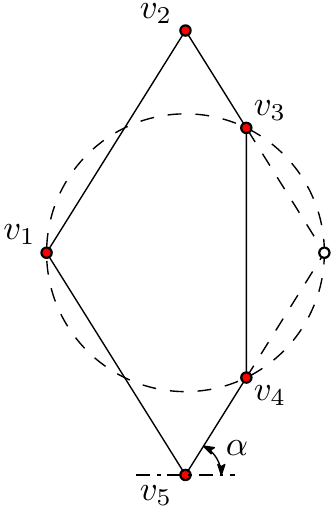}}%
  ~
  \subcaptionbox
  {\label{fig:rch_containment_red_points:2}
    The rectilinear convex hull is formed by three connected components: $v_2$, $v_5$, and an $L$-shaped region whose sides are parallel to the lines of $\o_\theta$.
    Points in the dashed regions lie in the intersection of opposite stabbing $\o_\theta$-wedges.
  }
  [.45\linewidth][c]{\includegraphics{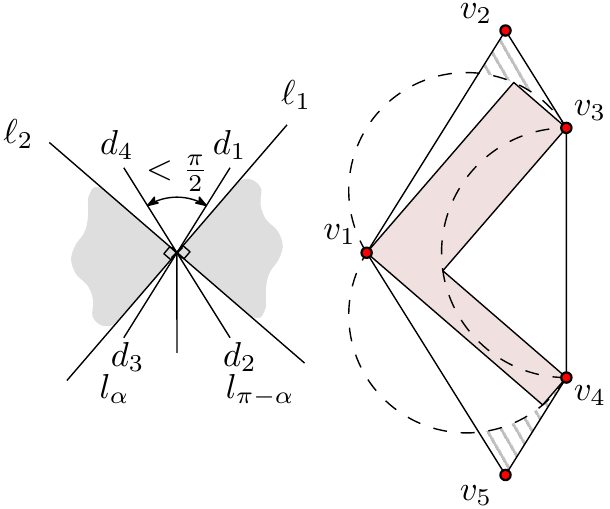}}%
  ~
  \subcaptionbox
  {\label{fig:rch_containment_red_points:3}
    The remaining points lie in the highlighted region.
  }
  [.25\linewidth][c]{\includegraphics{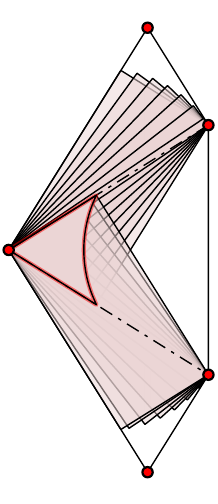}}%

  \caption{A bichromatic point set with $\Omega(n)$ angular intervals
    of containment: The set of red points.}
  \label{fig:rch_containment_red_points}
\end{figure}

The relevant rotation angles are again those in the interval $\varphi=[\frac{\pi}{2}-\alpha, \alpha]$.
Note that, since the direction of the edge $\overline{v_3 v_4}$ is parallel to the $Y$-axis, the observations we made about the construction described in Subsection~\ref{subsec:rch_separability_lower_bound} still hold: for any $\theta\in\varphi$ we have that $\ch[R]$ is stabbed only by $Q_2$- and $Q_4$-quadrants, and $\rcht[R]$ is formed by three connected components.
The relevant difference is the central component, which instead of a rectangle inscribed in $C$, is now an $L$-shaped region whose sides are parallel to the sides of $\o_\theta$; see \Cref{fig:rch_containment_red_points:2}.
While incrementing $\theta$ from $\frac{\pi}{2}-\alpha$ to $\alpha$, three of the vertices of this region are anchored at $v_1$, $v_3$, and $v_4$, while the remaining three vertices traverse the dashed semicircles in the counter-clockwise direction.
By intersecting the $L$-shaped regions for all the rotation angles in $\varphi$, we obtain the region highlighted in \Cref{fig:rch_containment_red_points:3}.
Note that any red point lying in this region is contained in $\rcht[R]$ for all $\theta\in\varphi$.
Hence, we may add as many red points as desired without affecting the construction.

The set of blue points is shown in \Cref{fig:rch_containment_blue_points}.
The points lie in the interior of the triangle with vertices $v_1, v_3, v_4$, on a circle with center on the middle point of the segment $\overline{v_3 v_4}$, and radius $\sfrac{d(v_3,v_4)}{2} - \varepsilon$, for $0 < \varepsilon < \sfrac{d(v_3,v_4)}{2}$.
The points are spread so at every $\theta\in\varphi$, at most a single blue point is \emph{not} contained in $\rcht[R]$.
As shown in \Cref{fig:rch_containment_rotation} (see the figures from left to right), while rotating the lines of $\o_\theta$ around the origin by incrementing $\theta$ from $\frac{\pi}{2}-\alpha$ to $\alpha$, the reflex vertex of the $L$-shaped region traverses the red circular arc in the clockwise direction.
Hence $\rcht[R]$ loses one blue point at a time, generating $\Omega(n)$ intervals of containment.

\begin{figure}[ht]
  \centering

  \subcaptionbox
  {\label{fig:rch_containment_blue_points}
    The points lie in the interior of the triangle with vertices $v_1, v_3, v_4$, on a circle with center on the middle point of the segment $\overline{v_3 v_4}$, and radius $\sfrac{d(v_3,v_4)}{2} - \varepsilon$.
  }
  [.4\linewidth][c]{\includegraphics{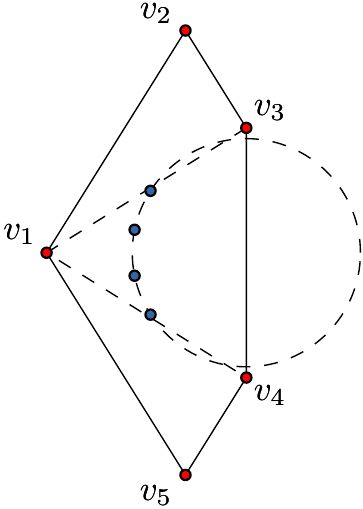}}%
  ~~
  \subcaptionbox
  {\label{fig:rch_containment_rotation}
    While incrementing $\theta$ from $\frac{\pi}{2}-\alpha$ to $\alpha$, $\rcht[R]$ loses one blue point at a time, generating $\Omega(n)$ angular intervals of containment.
  }
  [.6\linewidth][c]{%
    \includegraphics{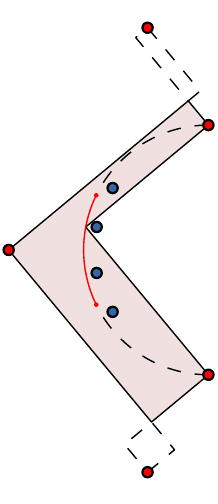}%
    \qquad%
    \includegraphics{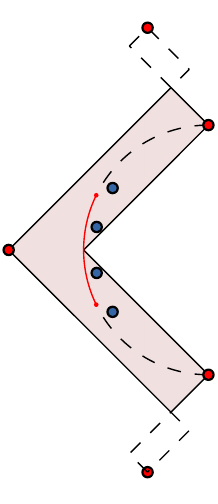}%
    \qquad%
    \includegraphics{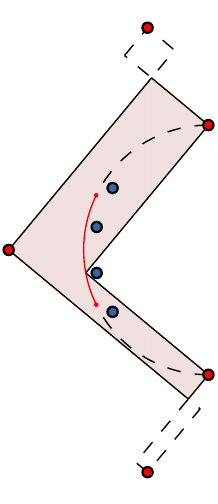}%
  }

  \caption
  {
    A bichromatic point set with $\Omega(n)$ angular intervals of containment: The set of blue points.
  }
  \label{fig:rch_containment_intervals}
\end{figure}

We summarize the lower bounds discussions of \Cref{subsec:rch_inclusion,subsec:rch_separability_lower_bound} in the following proposition.

\begin{proposition}
There exist disjoint sets $R$ and $B$ of red and blue points in the plane that induce $\Omega(n)$ intervals of $\theta$ in which either i) $\rchtR$ is $B$-free or ii) $\rchtR$ contains $\rchtB$, where $n = \vert R \vert + \vert B \vert$ and $R$ may have $O(1)$ points.
\end{proposition}

\section{Generalizations}
\label{sec:generalizations}

In this section we generalize the results from \Cref{sec:rch}.
First, in Subsection~\ref{subsec:och}, we consider the case in which the set $\o$ contains not only two lines, but $k \geq 2$ lines with arbitrary orientations.
In this setting the corresponding convex hull is known as the $\o$-convex hull~\cite{rawlins_thesis_1987}.
Then, in Subsection~\ref{subsec:obh}, we consider the case in which the set $\o$ of two orthogonal lines is substituted by a set $\ob$ of two lines that are not necessarily orthogonal to each other, but form
an angle $\beta\in(0,\pi)$.
In this setting the corresponding convex hull is known as the $\o_\beta$-convex hull~\cite{alegria_2018}.
We split the description of each generalization in three parts.
In the first part, we adapt the needed results from Subsection~\ref{subsec:maximal_wedges_and_arcs} to characterize the conditions in which a blue point is contained in the hull of the set of red points.
In the second part, we adapt the algorithm from Subsection~\ref{sec:rch_algorithm} to compute and maintain the set of blue points contained in the hull of the set of red points while we change the orientation of the lines of $\o$.
Finally, in the third part, we generalize the results from
Subsections~\ref{subsec:rch_separability_lower_bound} and~\ref{subsec:rch_inclusion} to bound the number of angular intervals of separability and containment between the hulls of the red and the blue point sets.

\subsection{The $\boldsymbol \o$-convex hull}
\label{subsec:och}

In this subsection we solve the following problem.

\begin{problem}\label{problem:och}
  Given a set $\o$ of orientations formed by $k \geq 2$ lines, compute the set of rotation angles for which the lines of $\o$ have to be simultaneously rotated counterclockwise around the origin, so the $\o$-convex hull of $R$ contains no points of $B$.
\end{problem}

For the sake of simplicity, throughout this subsection we
consider indices $i$ to be modulo $2k$.
We also assume that the $k \geq 2$ lines of $\o$ are labeled with $\ell_1,\dots,\ell_k$ so that $i<j$ implies that the orientation of $\ell_i$ is smaller than the orientation of $\ell_j$.
Let $r_i$ and $r_{i+k}$ denote the rays into which $\ell_i$ is split by the origin.
Given two indexes $i$ and $j$, we denote with $W_{i,j}$ the wedge spanned as we counterclockwise rotate $r_i$ anchored at the origin until we obtain $r_j$.
A \emph{$W_i^j$-wedge} is a translation of $W_{i,j}$.
We say that a $W_{i+1}^{i+k}$-wedge is an \emph{$\o$-wedge}, see \Cref{fig:o-convex_hull:1}.
The \emph{$\o$-convex hull} of a finite point set $P$, denoted with $\oh$, is the set
\[
  \oh =\mathbb{R}^{2} \setminus\bigcup_{i=1}^{2k}\mathcal{W}^i,
\]
where $\mathcal{W}^i$ denotes the union of all the $P$-free $W_{i+1}^{i+k}$-wedges, see \Cref{fig:o-convex_hull:2}.
Note that, as the rectilinear convex hull, the $\o$-convex hull of a finite point set is typically not convex, may be disconnected, and is orientation-dependent.
More details on these and other properties can be found in~\cite{fink_2004}.

\begin{figure}[ht]
  \centering
  
  \subcaptionbox
  {\label{fig:o-convex_hull:1}
    From top to bottom and left to right, the set $\o$ and the wedges $W_{i+1,i+k}$ for $i = 1,\ldots,2k$.
    Each wedge defines an infinite family of $\o$-wedges.
  }
  [0.4\linewidth][c] {
    \centering%
    \includegraphics[scale=0.8]{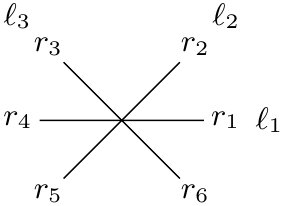}%
    \vspace{1em}%
    \linebreak[1]%
    \includegraphics[scale=0.7]{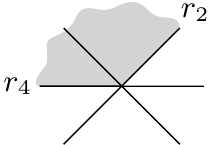}%
    \hfill%
    \includegraphics[scale=0.7]{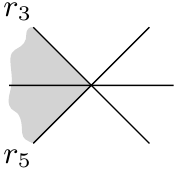}%
    \hfill%
    \includegraphics[scale=0.7]{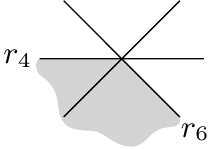}%
    \vspace{1em}%
    \linebreak[4]%
    \includegraphics[scale=0.7]{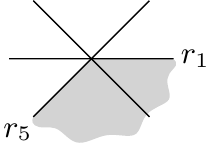}%
    \hfill%
    \includegraphics[scale=0.7]{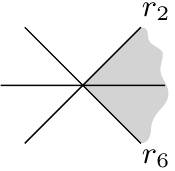}%
    \hfill%
    \includegraphics[scale=0.7]{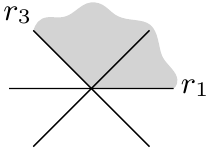}%
  }
  ~
  \subcaptionbox
  {\label{fig:o-convex_hull:2}
    On the left, the wedge $W_{3,5}$ is next to a $W_3^5$-wedge.
    On the right and from top to bottom, the wedge $W_{1,3}$ is next to a $W_1^3$-wedge, the wedge $W_{6,2}$ is next to a $W_6^2$-wedge, and the wedge $W_{5,1}$ is next to a $W_5^1$-wedge.
  }
  [0.55\linewidth][c]{
    \centering%
    \includegraphics{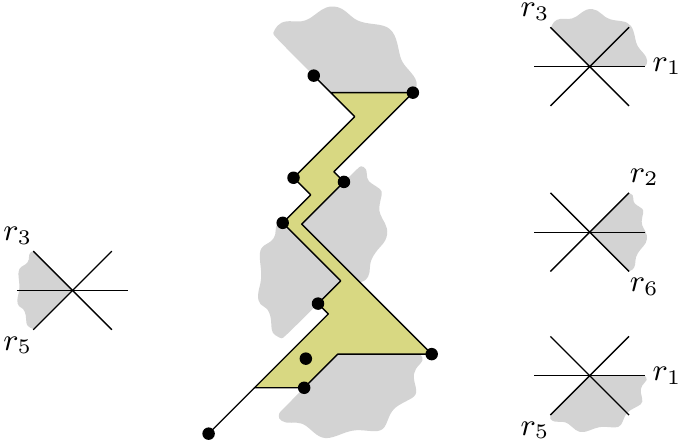}%
  }

  \caption
  {
    A set $\o$ of orientations with $k=3$ lines in Figure~\subref{fig:o-convex_hull:1}, and the $\o$-convex hull of a finite point set $P$ in Figure~\subref{fig:o-convex_hull:2}.
  }
  \label{fig:o-convex_hull}
\end{figure}

Let $\o_\theta$ denote the set of lines obtained after simultaneously rotating the lines of $\o$ counterclockwise around the origin by an angle of $\theta$.
We solve \cref{problem:och} by describing an algorithm to compute the (possibly empty) set of angular intervals of $\theta \in [0, 2\pi)$ for which the $\o_\theta$-convex hull of $R$ is $B$-free.
See \Cref{fig:o-convex_hull_components}.

\begin{figure}[ht]
  \centering

  {\includegraphics[page=1,scale=0.95]{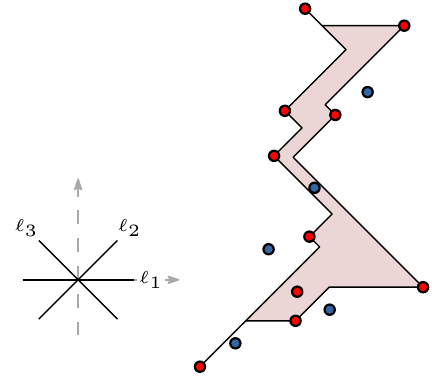}}%
  \quad%
  {\includegraphics[page=2,scale=0.95]{o-convex_hull__bichromatic}}
  \quad%
  {\includegraphics[page=3,scale=0.95]{o-convex_hull__bichromatic}}

  \caption
  {
    The sets $R$ and $B$, and the $\o_\theta$-convex hull of $R$ for three different values of $\theta$.
    In each figure, the set $\o$ is shown at the bottom left corner along with the coordinate axes, which are shown with dashed lines.
  }
  \label{fig:o-convex_hull_components}
\end{figure}

We start with the following generalization of \cref{prop:rch_inclusion}, which derives directly from the definition of $\o$-convex hull.

\begin{proposition}\label{prop:oh_inclusion}
  A point $x \in \R$ is contained in $\oh$ if, and only if, every $\o$-wedge with vertex on $x$ contains at least one point of $P$.
\end{proposition}

As in \cref{sec:rch}, we consider $\o$ to be not only a set of $k$ lines, but also the set of $2k$ rays in which the lines of $\o$ are split by the origin.
We generalize the definition of feasible maximal arc as follows.
Let $\alpha_i$ denote the size of the wedge $W_{i+1,i+k}$.
We denote with $\Theta = \operatorname{min}\{ \alpha_1, \ldots, \alpha_k\}$ the smallest angle among the sizes of the $\o$-wedges defined by the lines of $\o$.
We say that a maximal arc is \emph{feasible}, if it is induced by a maximal wedge with size at least~$\Theta$.
See \Cref{fig:oh_feasible_arcs_oset,fig:oh_feasible_arcs}.

\begin{figure}[ht]
  \centering

  \includegraphics[scale=0.8]{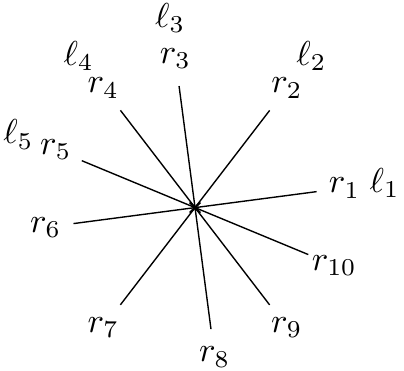}
  \\[1em]
  \includegraphics[scale=0.8]{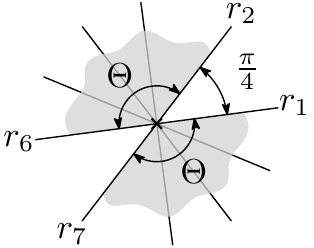}
  \hfill
  \includegraphics[scale=0.8]{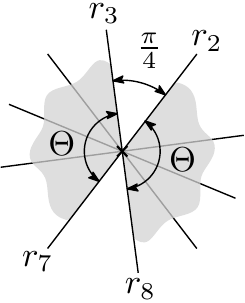}
  \hfill
  \includegraphics[scale=0.8]{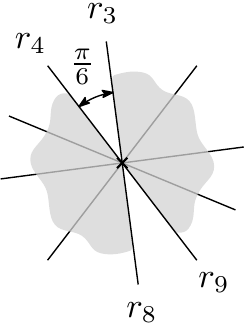}
  \hfill
  \includegraphics[scale=0.8]{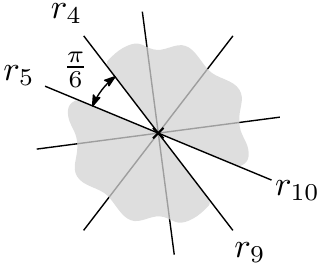}
  \hfill
  \includegraphics[scale=0.8]{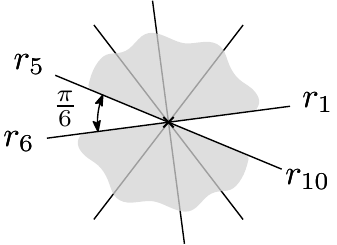}

  \caption
  {
    At the top, a set $\o$ with $k=5$ lines.
    At the bottom and from left to right, the wedges $W_{i+1,i+k}$ for $i = 1,\ldots,2k$.
    The first two figures show the wedges with the smallest size among all.
  }
  \label{fig:oh_feasible_arcs_oset}
\end{figure}

\begin{figure}[ht]
  \centering

  \subcaptionbox
  {\label{fig:rch_feasible_arcs:1}
    A feasible blue maximal arc is induced by a wedge with size at least $\Theta$.
  }
  {\includegraphics{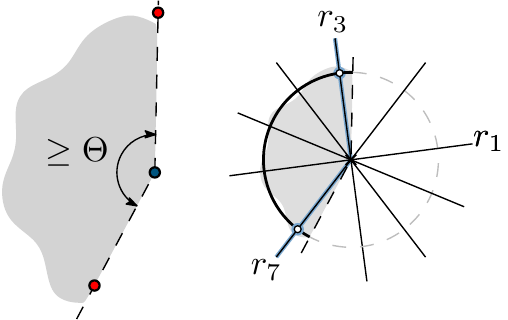}}%
  \qquad
  \subcaptionbox
  {\label{fig:rch_feasible_arcs:2}
    A blue maximal arc that is not feasible is induced by a wedge with size smaller than $\Theta$.
  }
  {\includegraphics{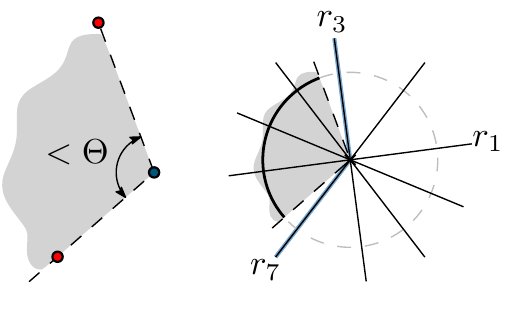}}

  \caption
  {
    A feasible arc for the set $\o$ of \Cref{fig:oh_feasible_arcs_oset}.
  }
  \label{fig:oh_feasible_arcs}  
\end{figure}

We generalize \cref{lem:rch_inclusion_arcs} as follows.

\begin{lemma}\label{lem:oh_inclusion_arcs}
  For any fixed value of $\theta$, a point $x \in \R$ is contained in the $\o_\theta$-convex hull of $P$ if, and only if, every feasible maximal arc of $x$ is intersected by a set $r_i,r_{i+1},\ldots,r_j$ of rays of $\o_\theta$ such that $j < i + k$.
\end{lemma}

\begin{proof}
  By a straightforward adaptation of the proof of \cref{lem:rch_inclusion_arcs}, we can show that every $\o_\theta$-wedge with vertex on $x$ contains at least one point of $P$ if, and only if, every feasible maximal arc of $x$ is intersected by a set $r_i,r_{i+1},\ldots,r_j$ of rays such that $j < i + k$.
  The key observation for this adaptation is that, if two rays $r_i$ and $r_j$ with $j = i + k$ intersect a maximal arc $a$ that is feasible, then $a$ is induced by a maximal wedge that contains an $\o_\theta$-wedge bounded by rays parallel to $r_i$ and $r_j$.
  The lemma follows from this fact and \Cref{prop:oh_inclusion}.
\end{proof}

In the following lemma we rephrase \cref{lem:oh_inclusion_arcs} to a bichromatic setting.

\begin{lemma}\label{lem:oh_bichromatic_inclusion_arcs}
  For every fixed value of $\theta$, a blue point $b\in B$ is contained in the $\o_\theta$-convex hull of $R$ if, and only if, every blue maximal arc of $b$ that is feasible is intersected by a set $r_i,r_{i+1},\ldots,r_j$ of rays of $\o_\theta$ such that $j < i + k$.
\end{lemma}

We now adapt the algorithm from Subsection~\ref{sec:rch_algorithm}.
The adaptation consists of four steps.
The first step is an additional preprocessing step in which we compute the angle $\Theta$.
The remaining steps are adaptations of those of the original algorithm.

\subparagraph*{Step 0. Computing the angle $\boldsymbol \Theta$.}

To compute the angle $\Theta$, we first sort the lines of $\o$ by orientation in increasing order in $O(k \log k)$ time and $O(k)$ space.
Then, we compute in $O(k)$ time the set of angles $\{ \alpha_1,\ldots,\alpha_k \}$.
We finally obtain $\Theta$ by keeping the smallest angle in the set.
Clearly, this step takes in $O(k \log k)$ time and $O(k)$ space.

\subparagraph*{Step 1. Computing the set of feasible maximal arcs.}

In this step we generalize the Step~1 of the original algorithm to compute the set of blue maximal arcs that are feasible.

We start by computing the set $\mathcal W$ of blue maximal wedges with size at least $\Theta$.
We proceed as follows.
A blue maximal wedge with size at least $\Theta$ is constrained to either the $X^+$ semiaxis, or to one of the $\sfrac{2\pi}{\Theta}$ directions defined by counterclockwise rotating $X^+$ by an integer multiple of $\Theta$.
By means of the algorithm we described in \Cref{subsec:maximal_wedges_and_arcs}, we compute the set of blue maximal wedges constrained to each one of these $\sfrac{2\pi}{\Theta} + 1$ directions.
From the resulting set of wedges, we obtain $\mathcal W$ by keeping those wedges whose size is at least $\Theta$.
By \cref{lem:restricted_maximal_wedges}, we have computed $\mathcal W$ in $O(\sfrac{1}{\Theta} \cdot n \log n)$ time and $O(\sfrac{1}{\Theta} \cdot n)$ space.
Moreover, note that $\mathcal W$ contains $O(\sfrac{1}{\Theta} \cdot n)$ wedges.

We now traverse $\mathcal W$, and process each wedge $w\in\mathcal{W}$ by transforming $w$ into a blue maximal arc that is feasible in $O(1)$ time.
Since each wedge is transformed into a single arc, there are $O(\sfrac{1}{\Theta} \cdot n)$ blue maximal arcs that are feasible.
Clearly, the time complexity of the whole step is $O(\sfrac{1}{\Theta} \cdot n \log n)$ time and $O(\sfrac{1}{\Theta} \cdot n)$ space.

\subparagraph*{Step 2. Computing the list of intersection events.}

In this step we generalize the Step~2 of the original algorithm to compute the sorted list of intersection events.
This step does not need to be modified; nevertheless, since we now have $O(\sfrac{1}{\Theta} \cdot n)$ intersection events, the original complexity is replaced by $O(\sfrac{1}{\Theta} \cdot n \log n)$ time and $O(\sfrac{1}{\Theta} \cdot n)$ space.

\subparagraph*{Step 3. Performing the angular sweep.}

Finally, in this step we generalize the Step~3 of the original algorithm to perform the angular sweep on the set of blue maximal arcs that are feasible.

The required adaptations are the following.
The set $B_\theta$ now denotes the set of blue points contained in the $\o_\theta$-convex hull of $R$.
The upper bound on $N_i$ is increased from four to $\sfrac{2\pi}{\Theta}$.
The variable $n_i$ now denotes the number of blue maximal arcs of $b_i$ that are intersected either by one ray of $\o_\theta$, or by a set $r_u,\ldots,r_v$ of rays of $\o_\theta$ such that $v < u + k$.
Following the condition from \cref{lem:oh_inclusion_arcs}, the array of $\vert B \vert$ Boolean flags used to encode the set $B_\theta$ has the $i$-th flag set to \texttt{True} if $n_i=N_i$, and to \texttt{False} if $n_i < N_i$.

The variable $\rho(a)$ now denotes the range of subindices of the rays of $\o_\theta$ intersecting the arc $a$.
Observe that $\rho(a)$ cannot be empty.
Suppose that $\rho(a) = (u,v)$, $u \leq v$, and, at an intersection event, a ray starts intersecting $a$.
Since the lines of $\o_\theta$ are labeled by increasing orientation and are rotated in the counter-clockwise direction, the range is thus increased to $(u-1,v)$.
If instead the ray stops intersecting $a$, then the range is reduced to $(u,v-1)$.
Finally, the queue $\mathcal{Q}$ now contains at most $2k$ angles instead of four.
Let $\theta_i$ denote the smallest counter clockwise rotation angle for which the ray $r_i\in\o_\theta$ passes over an endpoint of a blue maximal arc.
The queue contains the angles $\theta_i$, $1\leq i \leq 2k$, that are less than $\pi$.
Hence update operations on $\mathcal{Q}$ take $O(\log k)$ time.

The sweep is essentially performed in the same way as explained in the algorithm from Subsection~\ref{sec:rch_algorithm}.
There are slight modifications to the algorithm and an increment in the time and space complexities, consequence of having $O(\sfrac{1}{\Theta} \cdot n)$ intersection events and $O(k)$ angles in $\mathcal{Q}$.
Since the lines of $\o$ are already sorted by slope (refer to Step~0), the $O(n)$ time complexity of the initialization step is replaced by $O(\sfrac{1}{\Theta} \cdot n)$ time.
On the other hand, since $\o$ may not be symmetric, to perform the angular sweep we increment $\theta$ from $0$ to $\pi$ so the rays of $\o_\theta$ sweep all the directions of $\mathbb{S}^1$, and the endpoints of each blue maximal arc are touched by all the lines of $\o_\theta$.
Consider an intersection event $\theta$ for which a ray $r_j\in\o_\theta$ is passing over the endpoint of a maximal arc $a$ of a blue point $b_i$.
Assume that $\rho(a) = (u,v)$.
We process the intersection event as follows:
\begin{itemize}
\item
  If $r_j$ starts intersecting $a$ then $j = u-1$, so  we set $\rho(a) = (u-1,v)$ to add $j$ to $\rho(a)$.
  If instead $r_j$ stops intersecting $a$ then $j = v$, so we set $\rho(a) = (u,v-1)$ to remove $j$ from $\rho(a)$.
\item
  If $\rho(a)$ was changed then we update $n_i$ as follows.
  If $j$ was added to $\rho(a)$ and $v - j = k$ then we decrease $n_i$ by one.
  If instead $j$ was removed from $\rho(a)$ and $v - j = k - 1$ then we increase $n_i$ by one.
\item
  Finally, we update the set of Boolean flags that encode $B_\theta$, obtain the
  next intersection event, and update $\mathcal{Q}$ as explained in the
  algorithm of Subsection~\ref{sec:rch_algorithm}.
\end{itemize}

\begin{lemma}\label{lemma:och}
  The subset of blue points contained in the $\o_\theta$-convex hull of $R$ can be computed and maintained, while $\theta$ is increased from $0$ to $\pi$, in $O(\sfrac{1}{\Theta} \cdot N \log N)$ time and $O(\sfrac{1}{\Theta} \cdot N)$ space, where $N = \max \{ k, \vert R \vert + \vert B \vert \}$.
\end{lemma}

\begin{proof}
  As in the algorithm from Subsection~\ref{sec:rch_algorithm}, the most expensive step is the execution of the angular sweep (Step~3).
  Since each intersection event is processed in $O(1)$ time, the queue $\mathcal{Q}$ is updated in $O(\log k)$ time, and there are $O(\sfrac{1}{\Theta} \cdot n)$ intersection events, then the time and space complexities of Step~3 are $O(\sfrac{1}{\Theta} \cdot n \cdot \log k) = O(\sfrac{1}{\Theta} \cdot N \log N)$ time and $O(\sfrac{1}{\Theta} \cdot n) = O(\sfrac{1}{\Theta} \cdot N)$ space, where $N = \max \{ k, n = \vert R \vert + \vert B \vert \}$.
\end{proof}

From \cref{lemma:och} we obtain the main result of this subsection.

\begin{theorem}\label{thm:och_separability}
  Given two disjoint sets $R$ and $B$ of points in the plane and a set $\o$ of $k \geq 2$ lines with different orientations, the (possibly empty) set of angular intervals of $\theta \in [0,2\pi)$ for which the $\o_\theta$-convex hull of $R$ is $B$-free (i.e., the $\o_\theta$-convex hull of $R$ is a separator of $R$ and $B$) can be computed in $O(\sfrac{1}{\Theta} \cdot N \log N)$ time and $O(\sfrac{1}{\Theta} \cdot N)$ space, where $N = \max \{ k, \vert R \vert + \vert B \vert \}$.
\end{theorem}

There are a couple of remarks regarding the algorithm we described to prove \cref{thm:och_separability}.
First note that the time and space complexities are parametrized by both $\Theta$ and $k$.
If $\sfrac{1}{\Theta}$ is a constant value and $k$ is of the same order of magnitude than $\vert R \vert$ and $\vert B \vert$, then the complexities become $O(n \log n)$ time and $O(n)$ space.
These are the same complexities reported in \cref{thm:rch_separability} for the problem of separability by a rectilinear convex hull.
Second, as in \cref{thm:rch_separability}, the reported angular intervals are maximal and open, and the algorithm does not compute a separating $\o$-convex hull.
To actually compute an $\o$-convex hull separating $R$ from $B$, we first choose a rotation angle in an interval of separability, and then spend additional $O(\sfrac{1}{\Theta} \cdot N \log N)$ time~\cite{alegria_2020}.
Finally, we have an observation regarding the value of $k$, which derives from Observation~2 of~\cite{alegria_2020}.
To state the observation we first need to generalize the notion of stabbing quadrant we introduced in \cref{subsec:rch_separability_lower_bound}.

A connected component of the $\o$-hull of a finite point set $P$ is either (i) a single point of $P$, (ii) a polygonal chain of two segments parallel to the lines of $\o$ that connect two extremal points, or (iii) a closed polygon whose edges are parallel to the lines of $\o$.
The polygon may have ``degenerate edges'', which are line segments connecting its vertices with extremal points of $P$.
As with the rectilinear convex hull, the segments of the polygonal chains and the edges of the polygons are called the \emph{edges} of the $\o$-convex hull.
Each edge is contained in a ray of some $P$-free $\o$-wedge.
We say that such $\o$-wedge is \emph{stabbing $P$}.

\begin{observation}
  \label{obs:non_stabbing_wedges}
  Let $h$ be the number of edges of $\ch$.
  If $k$ is greater than $h$, then for any fixed value of $\theta$ there are $k - h$ lines in $\o_\theta$ that induce $\o_\theta$-wedges that do not stab $P$.
\end{observation}

\cref{obs:non_stabbing_wedges} implies that, if $k$ is greater than the number $h$ of edges of $\ch[R]$, then a separating $\o$-convex hull can be constructed using only $k-h$ of the lines of $\o$.
In \cref{fig:rch_disconnected} for example, any line added to the set $\o$ lying on the blue region (hence having an orientation greater than the orientation of $\ell_1$ and smaller than the orientation of $\ell_2$), induce $\o_\theta$-wedges that do not stab the point set $P$.

\subsubsection{Lower bound on the number of intervals of separability}
\label{subsubsec:oh_lowerbounds}

We now adapt the construction from Subsection~\ref{subsec:rch_separability_lower_bound} to obtain a $\Omega(n)$ bound on the number of intervals of separability.
For the sake of simplicity, we first describe the construction using a set $\o$ with $k=3$ lines with orientations $0$, $\frac{\pi}{3}$, and $\frac{2}{3}\pi$.
We later show how the construction can be extended to a set with more than three lines with arbitrary orientations.

Given two points $p$ and $q$, let $\ell_{pq}$ be the line through $p$ and $q$ directed from $p$ to $q$.
Let $C(p,q)$ be the circular arc spanned by $p$, $q$, and the angle $\frac{\pi}{3}$, that is contained in the right halfplane supported by $\ell_{pq}$.
We denote with $r(p,q)$ the radius of $C(p,q)$.
The convex hull of $R$ is again a rhombus whose diagonals are parallel to the coordinate axes; see \Cref{fig:och_adaptation_red_points:1}.
The points $v_2$ and $v_4$ lie outside the region bounded by $C(v_1,v_3) \cup C(v_3,v_1)$, so the interior angles of the rhombus at $v_2$ and $v_4$ are less than $\frac{\pi}{3}$.
Let $\alpha$ be the orientation of the line through $v_3$ and $v_4$.
For all $\theta$ in the interval $\varphi = [\frac{\pi}{3}-\alpha,\alpha]$, the direction of each edge of the rhombus lies in either $W_{2,3}$ or $W_{5,6}$; see \Cref{fig:och_adaptation_red_points:2}, left.
From this fact and a straightforward generalization of the arguments of Subsection~\ref{subsec:rch_separability_lower_bound} we have that, for all $\theta\in\varphi$, the $\o_\theta$-convex hull of $\{ v_1,\ldots,v_4 \}$ is formed by three connected components: the point $v_2$, the point $v_4$, and a rhombus inscribed in $C(v_1,v_3) \cup C(v_3,v_1)$ whose sides are parallel to $\ell_2$ and $\ell_3$; see \Cref{fig:och_adaptation_red_points:2}, right.
By intersecting all such rhombi for all the rotation angles in $\varphi$, we obtain the rhombus highlighted in \Cref{fig:och_adaptation_red_points:3}.
Note that any red point lying in this region is contained in the $\o_\theta$-convex of $R$ for all $\theta\in\varphi$.
Hence, we may add as many red points as desired without affecting the construction.

\begin{figure}[ht]
  \centering
  \subcaptionbox
  {\label{fig:och_adaptation_red_points:1}
    $\ch[R]$ is a rhombus whose diagonals are parallel to the coordinate axes.
  }
  [.25\linewidth][c]{\includegraphics[page=1]{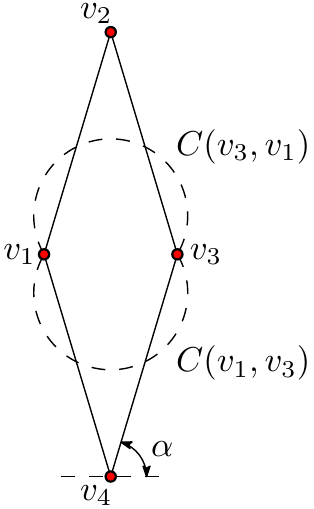}}
  ~
  \subcaptionbox
  {\label{fig:och_adaptation_red_points:2}
    The $\o$-convex hull has three connected components: the points $v_2$ and $v_4$, and a rhombus whose sides are parallel to the lines of $\o_\theta$.
  }
  [.45\linewidth][c]{\includegraphics[page=2]{separability_inclusion_oh}}
  ~
  \subcaptionbox
  {\label{fig:och_adaptation_red_points:3}
    The remaining points lie in the highlighted rhombus.
  }
  [.2\linewidth][c]{\includegraphics[page=3]{separability_inclusion_oh}}
  
  \caption
  {
    Adapting the construction of \Cref{subsec:rch_separability_lower_bound}: The set of red points.
  }
  \label{fig:och_adaptation_red_points}
\end{figure}

The set of blue points is shown in \Cref{fig:och_adaptation_blue_points:1,fig:och_adaptation_blue_points:2}.
The points lie in the interior of $\ch[R]$, on a circle concentric to
$C(v_1,v_3)$ with radius $r(v_1,v_3) - \varepsilon$, for $0 < \varepsilon < r(v_1,v_3)$.
Note that the observations and lemmas from Subsection~\ref{subsec:rch_separability_lower_bound} can be applied to this construction with minor and straightforward modifications.
Therefore, the bichromatic point set has $\Omega(n)$ angular intervals of separability.

\begin{figure}[ht]
  \centering

  \subcaptionbox
  {\label{fig:och_adaptation_blue_points:1}
    The points lie on a circle concentric to $C(v_3,v_1)$ and radius $r(v_3,v_1) - \varepsilon$.
  }
  [.18\linewidth][c]{\includegraphics[page=1]{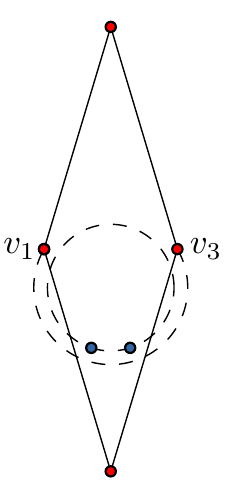}}
  ~
  \subcaptionbox
  {\label{fig:och_adaptation_blue_points:2}
    While incrementing $\theta$, the $\o_\theta$-convex hull of $R$ captures one blue point at a time, generating $\Omega(n)$ angular intervals of separability.
  }
  [.5\linewidth][c]
  {%
    \includegraphics[page=2]{separability_intervals_oh}%
    ~
    \includegraphics[page=3]{separability_intervals_oh}%
    ~
    \includegraphics[page=4]{separability_intervals_oh}%
  }
  ~
  \subcaptionbox
  {\label{fig:och_adaptation_blue_points:3}
    When $k>3$, for the construction we use a pair of consecutive lines of $\o$.
  }
  [.25\linewidth][c]{\includegraphics[page=5,scale=0.8]{separability_intervals_oh}}

  \caption
  {
    Adapting the construction of \Cref{subsec:rch_separability_lower_bound}: The set of blue points is described in Figures \subref{fig:och_adaptation_blue_points:1} and \subref{fig:och_adaptation_blue_points:2}, and a generalization to a set $\o$ with more than three lines in Figure \subref{fig:och_adaptation_blue_points:3}.
  }
  \label{fig:och_adaptation_blue_points}
\end{figure}

Consider now a set $\o$ of $k>3$ lines with arbitrary orientations.
To extend the previous construction to this case, first choose any pair of consecutive lines $\ell_i$ and $\ell_{i+1}$ in $\o$.
Then create the point sets as described above, using a rhombus $\mathcal{R}$ such that the internal angles at two opposite vertices are less than the size of the wedge $W_{i,i+1}$.
As shown in \Cref{fig:och_adaptation_blue_points:3}, the directions of the edges of $\mathcal{R}$ lie either in $W_{i,i+1}$ or $W_{i+k,i+k+1}$.
It is thus not hard to see that the arguments above still hold, so the construction has $\Omega(n)$ angular intervals of separability.

\subsubsection{Inclusion detection}
\label{subsec:o_hull_inclusion}

It is not hard to see that, with minor modifications, the statement and proof of \cref{lem:rch_containment} can be generalized to $\o$-convexity.
Hence, in the algorithm we described to prove \cref{thm:och_separability},
if there is a value of $\theta$ for which all the Boolean flags of the set that encodes $B_\theta$ are \texttt{True}, then the $\o_\theta$-convex hull of $B$ is contained in the $\o_\theta$-convex hull of $R$.
We obtain the following theorem as a consequence of \cref{lemma:och}.

\begin{theorem}\label{thm:och_containment}
 Given two disjoint sets $R$ and $B$ of points in the plane, the (possibly empty) set of angular intervals of $\theta \in (0,2\pi]$ for which the $\o_\theta$-convex hull of $B$ is contained in the $\o_\theta$-convex hull of~$R$ can be computed in $O(\sfrac{1}{\Theta} \cdot N \log N)$ time and $O(\sfrac{1}{\Theta} \cdot N)$ space, where $N = \max \{ k, \vert R \vert + \vert B \vert \}$.
\end{theorem}

We now combine the constructions we described in Subsections~\ref{subsec:rch_inclusion}~and~\ref{subsubsec:oh_lowerbounds} to obtain a point set with $\Omega(n)$ angular intervals of containment.
We assume that $\o$ contains $k=3$ lines with orientations $0$, $\frac{\pi}{3}$, and $\frac{2}{3}\pi$.
The construction can be extended to a set with more than three lines with arbitrary orientations, in a similar way as we described in Subsection~\ref{subsubsec:oh_lowerbounds}.
The point set is illustrated in \Cref{fig:och_inclusion_intervals}.
The adaptation is based on the rhombus we used in the construction from Subsection~\ref{subsubsec:oh_lowerbounds}; refer again to \Cref{fig:och_adaptation_red_points}.
Three red points lie on vertices of the rhombus, and two red points on the intersections between the rhombus and a vertical line.
The line is chosen such that the region bounded by $C(v_3,v_4) \cup \overline{v_3 v_4}$ does not contain $v_1$.
Note that any red point lying in this region highlighted in red is contained in the $\o_\theta$-convex of $R$.
Hence, we may add as many red points as desired without affecting the construction.
The blue points lie in the interior of the triangle with vertices $v_1$,$v_3$, and $v_4$, on a circle concentric to $C(v_3,v_4)$ with radius $r(v_3,v_4) - \varepsilon$, for $0 < \varepsilon < r(v_3,v_4)$.
As in Subsections~\ref{subsec:rch_inclusion} the points are spread so that, while rotating the lines of $\o$ around the origin, the $\o$-convex hull of $R$ loses one blue point at a time.
From similar arguments as those made on Subsections~\ref{subsec:rch_inclusion}~and~\ref{subsubsec:oh_lowerbounds}, the bichromatic point set has $\Omega(n)$ angular intervals of containment.

\begin{figure}[ht]
  \centering

  \subcaptionbox
  {\label{fig:och_inclusion_intervals:1}Construction of the point set.}
  [.2\linewidth][c]{\includegraphics[page=1]{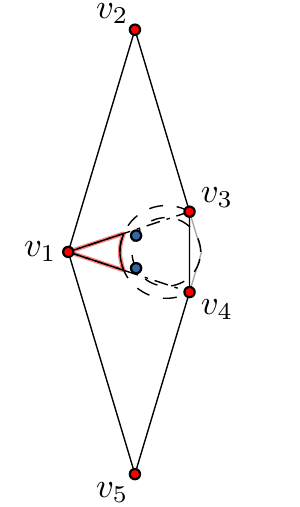}}
  ~
  \subcaptionbox
  {\label{fig:och_inclusion_intervals:2}
    While rotating the lines of $\o$ around the origin, the $\o$-convex hull of $R$ loses one point at a time, generating $\Omega(n)$ angular intervals of containment.
  }
  [.75\linewidth][c]{%
    \includegraphics[page=2]{inclusion_intervals_oh}%
    \quad
    \includegraphics[page=3]{inclusion_intervals_oh}%
    \quad
    \includegraphics[page=4]{inclusion_intervals_oh}%
  }

  \caption
  {
    A bichromatic point set with $\Omega(n)$ angular intervals of containment.
    The red points lie on the boundary of a rhombus, and in the interior of the red highlighted region.
    The blue points lie on a circle concentric to $C(v_3,v_4)$ with radius $r(v_3,v_4) - \varepsilon$.
  }
  \label{fig:och_inclusion_intervals}
\end{figure}

We summarize the lower bounds discussions of \Cref{subsec:o_hull_inclusion,subsubsec:oh_lowerbounds} in the following proposition.

\begin{proposition}
There exist disjoint sets $R$ and $B$ of red and blue points in the plane that induce $\Omega(n)$ intervals of $\theta$ in which either i) the $\o$-convex hull of $R$ is $B$-free or ii) the $\o$-convex hull of $R$ contains the $\o$-convex hull of $B$, where $n = \vert R \vert + \vert B \vert$ and $R$ may have $O(1)$ points.
\end{proposition}

\subsection{The $\boldsymbol{\ob}$-convex hull}
\label{subsec:obh}

Let $\ob$ denote a set of orientations formed by two lines with orientations $0$ and $\beta$.
An \emph{$\ob$-quadrant} is one of the four open wedges that result from subtracting the lines of $\ob$ from the plane. The $\ob$-convex hull of a finite point set $P$, denoted with $\obh$, is the set
\[
  \obh = \mathbb{R}^{2} \setminus \bigcup_{q\in\mathcal{Q}}q,
\]
where $\mathcal{Q}$ denotes the set of all $P$-free $\ob$-quadrants of the plane~\cite{alegria_2018}.
In this subsection we solve the following problem.

\begin{problem}
  \label{problem:obh}
  Compute the set of values of $\beta\in(0,\pi)$ for which the $\ob$-convex hull of $R$ contains no points of $B$.
\end{problem}

For the sake of simplicity, throughout this section we assume the set $R \cup B$ not only contains no three points on a line, but also no pair of points on a horizontal line.
To solve \Cref{problem:obh}, we adapt the results from \cref{sec:rch} to find the values of $\beta \in (0,\pi)$ for which $\obh[R]$ is $B$-free; see \Cref{fig:obh_bichromatic_inclusion}.
We start with the adaptation of \cref{prop:rch_inclusion}, which derives directly from the definition of $\ob$-convex hull.

\begin{figure}[ht]
  \centering

  \subcaptionbox
  {\label{fig:obh_bichromatic_inclusion:1}
    $\obh[R]$ is $B$-free, and is formed by the red points and a line segment with orientation $\beta$.
  }
  [0.3\linewidth][c]
  {\includegraphics[scale=1,page=1]{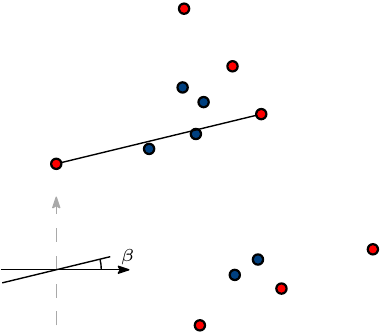}}
  ~
  \subcaptionbox
  {\label{fig:obh_bichromatic_inclusion:2}
    $\obh[R]$ contains some points of $B$, and is formed by three connected components.
  }
  [0.3\linewidth][c]
  {\includegraphics[scale=1,page=2]{bichromatic_inclusion_obh}}
  ~
  \subcaptionbox
  {\label{fig:obh_bichromatic_inclusion:3}
    $\obh[R]$ contains all the point of $B$, and is formed by two connected components.
  }
  [0.3\linewidth][c]
  {\includegraphics[scale=1,page=3]{bichromatic_inclusion_obh}}

  \caption
  {
    The sets $R$, $B$, and $\obh[R]$ for three different values of $\beta\in (0,\pi)$.
    The set $\ob$ and the coordinate axes are shown in the bottom-left corner of each figure.
  }
  \label{fig:obh_bichromatic_inclusion}
\end{figure}

\begin{proposition}\label{prop:obh_inclusion}
A point $x\in\R$ is contained in $\obh$ if, and only if, every $\ob$-quadrant with vertex on $x$ contains at least one point of $P$.
\end{proposition}

As in \cref{sec:rch}, we consider $\ob$ to be not only a set of two lines, but also the set of four rays in which the lines are split by the origin. We adapt the definition of feasible maximal arc as follows: we say that a maximal arc is \emph{feasible} if it is induced by a $P$-free maximal wedge constrained to either the $X^+$ or the $X^-$ semiaxis. See \Cref{fig:obh_feasible_arcs}. 

\begin{figure}[ht]
  \centering

  \subcaptionbox
  {\label{fig:obh_feasible_arcs:1}A maximal arc that is feasible.}
  [0.4\linewidth][c]{\includegraphics{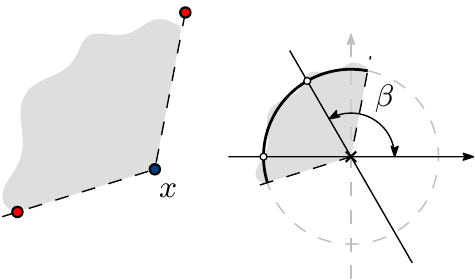}}%
  \qquad
  \subcaptionbox
  {\label{fig:obh_feasible_arcs:2}A maximal arc that is not feasible.}
  [0.4\linewidth][c]{\includegraphics{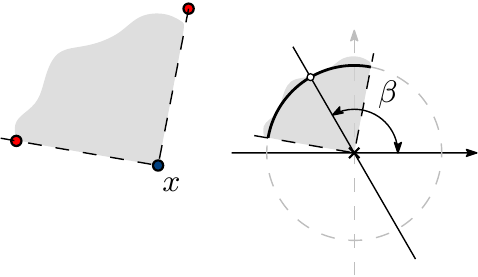}}

  \caption
  {
    Maximal arcs of a point $x\in\R$.
  }
  \label{fig:obh_feasible_arcs}
\end{figure}

We now adapt \cref{lem:rch_inclusion_arcs} as follows.

\begin{lemma}\label{lem:obh_inclusion_arcs}
For any fixed value of $\beta$, a point $x \in \R$ is contained in $\obh$ if, and only if, every feasible maximal arc of $x$ is intersected by a single ray of $\ob$.
\end{lemma}

\begin{proof}
We show that every $\ob$-quadrant with vertex on $x$ contains at least one point of $P$ if, and only if, every feasible maximal arc of $x$ is intersected by a single ray of $\ob$. The lemma follows from this fact and \Cref{prop:obh_inclusion}.

For any fixed value of $\beta$ there is an affine transformation that maps horizontal lines to horizontal lines, and lines with orientation $\beta$ to vertical lines~\cite[Section 2.5]{rawlins_thesis_1987}.
Let $\o^\prime_\beta$ and~$x^\prime$ denote the set and the point obtained after applying the transformation to $\ob$ and $x$, respectively. Assume without loss of generality that $x^\prime$ lies on the origin. The proof follows by observing that (i) the lines of $\o^\prime_\beta$ coincide with the coordinate axes, (ii) every maximal wedge with vertex on $x^\prime$ that induces a feasible maximal arc contains either the $X^+$ or the~$X^-$ semiaxes, and (iii) by similar arguments to those we used to prove \cref{lem:rch_inclusion_arcs}, such wedge contains a second semiaxis if, and only if, $x^\prime$ is contained in $\obh$.
\end{proof}

We rephrase \Cref{lem:obh_inclusion_arcs} to a bichromatic setting as follows.

\begin{lemma}\label{lem:obh_bichromatic_inclusion_arcs}
  A blue point $b\in B$ is contained in $\obh[R]$ if, and only if, every blue maximal arc of $b$ that is feasible is intersected by a single ray of $\ob$.
\end{lemma}

Let $\mathcal{H}(R)$ denote the narrowest horizontal corridor enclosing $R$. Consider a blue point~$b$ lying outside $\mathcal{H}(R)$. Note that $b$ has a single maximal arc that is feasible, since $b$ is the vertex of a single $R$-free maximal wedge constrained to both the $X^+$ and the $X^-$ semiaxis. Such arc is intersected by both lines of $\ob$ for all $\beta \in (0,\pi)$.
Hence, by \Cref{lem:obh_bichromatic_inclusion_arcs}, the point~$b$ is not contained in $\obh[R]$ for all $\beta \in (0,\pi)$. See \Cref{fig:obh_outside_corridor}.

\begin{figure}[ht]
  \centering
  \includegraphics[page=1,scale=0.9]{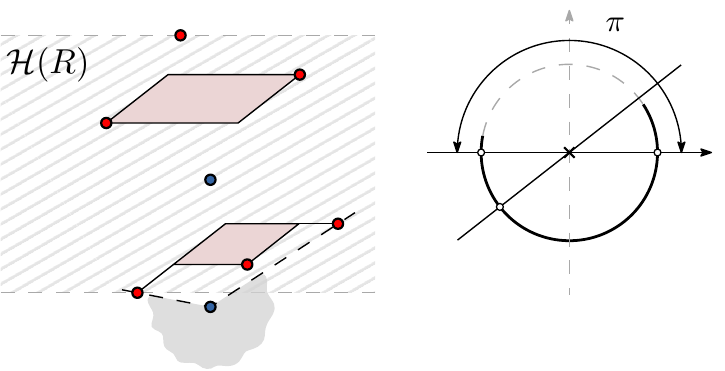}

  \caption{A blue point lying outside $\mathcal{H}(R)$ has a single feasible maximal arc, which is induced by a wedge constrained to both semiaxes of $X$. The point is not contained in $\obh[R]$ for all $\beta \in (0,\pi)$.}
  \label{fig:obh_outside_corridor}
\end{figure}

Consider now that $b$ is contained in $\mathcal{H}(R)$.
In this case $b$ has two maximal arcs that are feasible, since $b$ is the vertex of a single $R$-free maximal wedge constrained to the $X^+$ semiaxis, and a single $R$-free maximal wedge constrained to the $X^-$ semiaxis.
Such arcs are intersected by both lines of $\ob$ for values of $\beta$ in two angular intervals $(0,\beta_1)$ and $(\beta_2,\pi)$ for some angles $\beta_1,\beta_2\in (0,\pi)$ such that $0<\beta_1 <\beta_2 <\pi$. Hence, by \Cref{lem:obh_bichromatic_inclusion_arcs}, the point $b$ is not contained in $\obh[R]$ for values of $\beta$ in the intervals $(0,\beta_1)$ and $(\beta_2,\pi)$, and it is contained in $\obh[R]$ values of $\beta$ in the interval $(\beta_1,\beta_2)$.
See \Cref{fig:obh_inside_corridor}.

\begin{figure}[ht]
  \centering

  \subcaptionbox
  {\label{fig:obh_inside_corridor:1}
    The point is not contained in $\obh[R]$ for any $\beta \in (0,\beta_1) \cup (\beta_2,\pi)$.
  }
  {\includegraphics[page=2,scale=0.9]{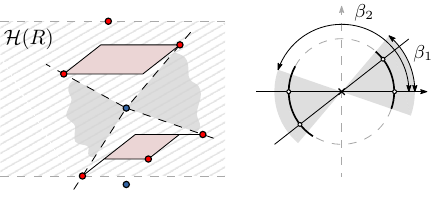}}
  \hfill
  \subcaptionbox
  {\label{fig:obh_inside_corridor:2}
    The point is contained in $\obh[R]$ for any $\beta \in (\beta_1,\beta_2)$.
  }
  {\includegraphics[page=3,scale=0.9]{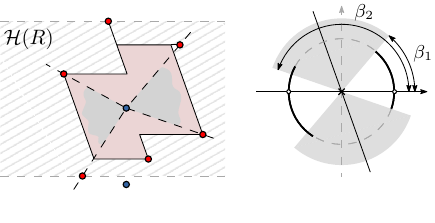}}

  \caption{
    A blue point lying inside $\mathcal{H}(R)$ has two feasible maximal arcs induced by two maximal wedges: one constrained to the  $X^+$ semiaxis, and the second one constrained to the $X^-$ semiaxis.
  }
  \label{fig:obh_inside_corridor}
\end{figure}

Let $b_1,\ldots,b_{\vert B \vert}$ be the points of $B$, and $S_i$ denote the set of angular intervals of $\beta$ for which $b_i$ is not contained in $\obh[R]$.
From the discussion above we have that $S_i$ contains either one or two angular intervals.
For some small enough $\varepsilon$, the intervals contain either the angle $\varepsilon$, the angle $\pi - \varepsilon$, or both.
Hence, the set $\underset{i=1}{\overset{\vert B \vert}{\bigcap}} \, S_i$ of angular intervals of $\beta$ for which $\obh[R]$ is $B$-free 
consists of at most two intervals.

\begin{theorem}\label{thm:obh_separability}
  Given two disjoint sets $R$ and $B$ of points in the plane, there are at most two open angular intervals of $\beta\in (0,\pi)$ where $\obh[R]$ is $B$-free.
  These intervals can be computed in $O(n\log n)$ time and $O(n)$ space, where $n=\vert R\vert +\vert B\vert$.
\end{theorem}

\begin{proof}
  By means of the algorithm we described in the proof of \Cref{lem:restricted_maximal_wedges}, we compute the set of blue maximal wedges that are constrained to either the $X^+$ or the $X^-$ semiaxis in $O( n \log n)$ time and $O(n)$ space.
  We then transform in $O(n)$ time the resulting set of maximal wedges into a set of maximal arcs that are feasible, as described in the proof of \Cref{lem:feasible_maximal_arcs}.
  Next, we transform each of such arcs into an angular interval in $O(1)$ time, as we described in Step 2 of the algorithm from \Cref{sec:rch_algorithm}.
  In $O(n)$ time, we use these intervals to compute the set $S_i$ of angular intervals for which a blue point $b_i$ is not contained in $\obh[R]$, for all $1 \leq i \leq \vert B \vert $.
  We finally compute the angular intervals where $\obh[R]$ is $B$-free in $O(n)$ time, by computing the set $\underset{i=1}{\overset{\vert B \vert}{\bigcap}} \, S_i$.
\end{proof}

As discussed above, the values of $\beta$ where a blue point is contained in $\obh[R]$ form at most a single angular interval.
We obtain the following result from this fact and similar arguments to those we use to prove \Cref{thm:obh_separability}.

\begin{theorem}\label{thm:obh_containment}
  Given two disjoint sets $R$ and $B$ of points in the plane, there is at most one open angular interval of $\beta \in (0,\pi)$ where $\obh[B]$ is contained in $\obh[R]$.
  This interval can be computed in $O(n \log n)$ time and $O(n)$ space, where $n = \vert R \vert + \vert B \vert$.
\end{theorem}

We finish this section with the following two remarks regarding \Cref{thm:obh_separability,thm:obh_containment}.
First, for a fixed value of $\beta$, the $\ob$-convex hull of a finite point set can be computed in $O(n \log n)$ time and $O(n)$ space~\cite{alegria_2018}.
Therefore, we need to spend an additional $O(n \log n)$ time and $O(n)$ space to compute the actual monochromatic $\ob$-convex hull in \Cref{thm:obh_separability}, or the $\ob$-convex hull of $R$ or $B$ in \Cref{thm:obh_containment}.
Second, since there is a constant number of angular intervals of $\beta$ where $\obh[R]$ is $B$-free or $\obh[B]$ is contained in $\obh[R]$, one may think that these intervals can be computed in $O(n)$ time.
Nevertheless, we show in \Cref{sec:lower_bounds} that the best possible time bound is actually $\Omega(n \log n)$.

\section{Lower bounds separation and inclusion detection problems}
\label{sec:lower_bounds}

In this section we consider the lines of $\o$ to be fixed, and we prove an $\Omega(n\log n)$ time lower bound in the algebraic computation tree model for the following problems:

\begin{problem*}[Rectilinear Convex Hull Separability Detection, RH-SD]
  Given two disjoint sets of $n$ red and $n$ blue points in the plane, decide if no blue point is contained in the rectilinear convex hull of the red point set.
\end{problem*}

\begin{problem*}[Rectilinear Convex Hull Containment Detection, RH-CD]
  Given two disjoint sets of $n$ red and $n$ blue points in the plane, decide if all the blue points are contained in the rectilinear convex hull of the red point set.
\end{problem*}

\begin{problem*}[Rectilinear Convex Hull Point Inclusion, RH-PI]
  Given two disjoint sets of $n$ red and $n$ blue points in the plane, compute the subset of blue points contained in the rectilinear convex hull of the red point set.
\end{problem*}

Note that these problems are particular cases of those we studied in \Cref{sec:rch,subsec:och,subsec:obh}.
For the problems related to the rectilinear convex hull (\Cref{sec:rch}) and the $\o$-convex hull (\Cref{subsec:och}), the set $\o_\theta$ is fixed to the case where $\theta$ is a constant value and contains $k=2$ orthogonal lines.
For the problem related to the $\ob$-convex hull (\Cref{subsec:obh}), the set $\ob$ is fixed to the case where $\beta=\frac{\pi}{2}$.
The $\Omega(n\log n)$ time lower bounds thus imply that the time complexities reported in \Cref{lemma:rch,thm:rch_separability,thm:rch_containment,thm:obh_separability,thm:obh_containment} are the best possible.

We first prove the lower bounds for the RH-SD and the RH-CD problems.
The proofs are by reduction from the following auxiliary problems.

\begin{problem*}[$\varepsilon$-Closeness]
  Given a set $x_1,\ldots,x_n$ and $\varepsilon > 0$ of $n+1$ real numbers, decide whether any two numbers $x_i$ and $x_j$ ($i \neq j$) are at distance less than $\varepsilon$ from each other.
\end{problem*}

Given a set $x_1,\ldots,x_n$ of real numbers, we say that two numbers $x_i$ and $x_j$ are \emph{consecutive} if $x_i < x_j$, and there is no $k$ such that $x_i < x_k < x_j$.

\begin{problem*}[Complement-Greater-or-Equal, CGE]
  Given a set $x_1,\ldots,x_n$ and $\varepsilon > 0$ of $n+1$ real numbers, decide whether the maximum distance between consecutive numbers is less than~$\varepsilon$.
\end{problem*}

The problems $\varepsilon$-Closeness and CGE have an $\Omega(n\log n)$ time lower bound in the algebraic computation tree model~\cite{arkin_2006,preparata_1985}.
The reductions from these problems are based on a construction we describe next.
We transform a set $x_1,\ldots,x_n$ and $\varepsilon > 0$ of $n+1$ real numbers, into two disjoint sets $R$ and $B$ of $2n+2$ red and $n-1$ blue points, such that some blue point is contained in the rectilinear convex hull of the red point set if, and only if, the distance between a pair of consecutive numbers is less than $\varepsilon$.

Let $x_{min} = \operatorname{min} \{x_1,\ldots,x_n\}$, $x_{max} = \operatorname{max} \{x_1,\ldots,x_n\}$, and $m$, $M$ be two real numbers such that $m \ll x_{min}$ and $M \gg x_{max}$.
The set $B$ is produced by transforming the set $\{ x_1,\ldots,x_n \} \setminus \{ x_{max} \}$ of real numbers into the set
\[
  \{ b_i = (x_i, x_i) \; \vert \; 1 \leq i \leq n \} \setminus \{ b_{max} = (x_{max},  x_{max}) \}
\]
of $n-1$ blue points on the line $\ell$ with equation $y = x$.
The set $R$ is produced by transforming the set $x_1,\ldots,x_n$ of real
numbers into the set
\[
  \{ r^{\tiny{+}}_i = (x_i - \varepsilon, x_i) \; \vert \; 1 \leq i \leq
  n \}
\]
of $n$ red points on the line $\ell^{\tiny{+}}$ with equation $y = x + \varepsilon$, the set
\[
  \{ r^{\tiny{-}}_i = (x_i, x_i - \varepsilon) \; \vert \; 1 \leq i
  \leq n \}
\]
of $n$ red points on the line $\ell^{\tiny{-}}$ with equation $y = x - \varepsilon$, and the points $r_m = (m, m), r_M = (M, M)$ on $\ell$.
See \Cref{fig:lower_bound_construction}.

\begin{figure}[ht]
  \centering
  \includegraphics{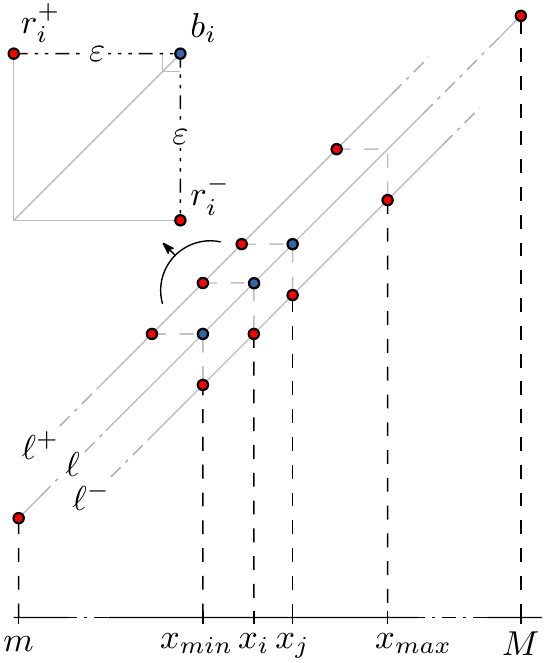}

  \caption
  {
    Transforming the set $x_1,\ldots,x_n$ and $\varepsilon>0$ of real numbers into two disjoint sets of red and blue points.
    The blue points lie on the line $\ell$ ($y=x$).
    The red points lie on the lines $\ell^{\tiny{+}}$ ($y=x+\varepsilon$) and $\ell^{\tiny{-}}$ ($y=x-\varepsilon$).
  }
  \label{fig:lower_bound_construction}
\end{figure}

Let $x_i$ and $x_j$ be two numbers in the set $x_1,\ldots,x_n$ such that $x_i < x_j$.
Consider the four different quadrants whose vertices lie on the blue point $b_i$.
Remember that a quadrant is an open region.
Note that the $Q_1$-quadrant contains the red point $r_M$ and the $Q_3$-quadrant contains the red point $r_m$.
If the distance between $x_i$ and $x_j$ is \emph{less than} $\varepsilon$, then the $Q_2$-quadrant contains the red point $r^{\tiny{+}}_j$ and the $Q_4$-quadrant contains the red point $r^{\tiny{-}}_j$.
By \Cref{prop:rch_inclusion}, in such case the blue point $b_i$ is strictly contained in $\rch[R]$.
If instead the distance between $x_i$ and $x_j$ is \emph{at least} $\varepsilon$, then both the $Q_2$-quadrant and the $Q_4$-quadrant are $R$-free, hence $b_i$ is not strictly contained in $\rch[R]$.
See \Cref{fig:lower_bound_rch}.

\begin{figure}[ht]
  \centering
  \includegraphics{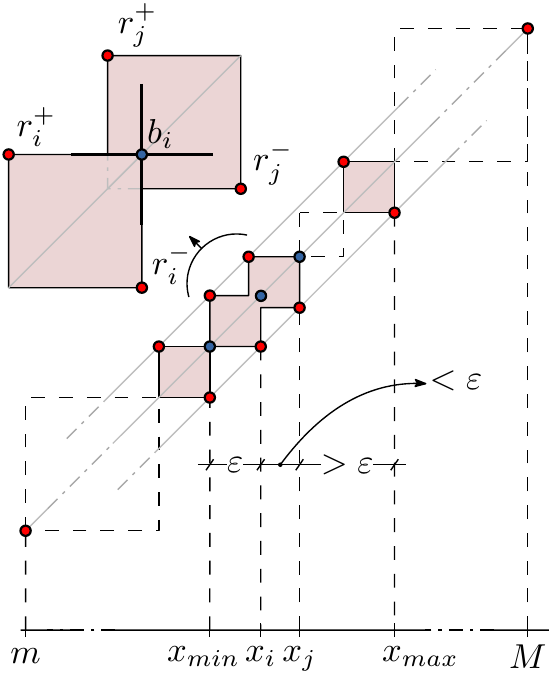}
  \caption
  {
    All the points of $R$ are vertices of $\rch[R]$.
    The points $r_m$ and $r_M$ are two singleton connected components, whereas any other connected component of $\rch[R]$ is an orthogonal polygon whose sides are parallel to the coordinate axes.
    Since the distance between $x_{i}$ and $x_j$ is less than $\varepsilon$, then $b_i$ is strictly contained in $\rch[R]$.
  }
  \label{fig:lower_bound_rch}
\end{figure}

\begin{lemma}\label{lem:construction}
  The construction described above transforms the set $x_1,\ldots,x_n$ and $\varepsilon>0$ of $n+1$ real numbers into two disjoint sets $R$ and $B$ of $2n+2$ red and $n-1$ blue points in $O(n)$ time and space.
  A pair of numbers $x_i$ and $x_j$, $x_i < x_j$, are at distance less than $\varepsilon$ if, and only if, the blue point $b_i$ is strictly contained in $\rch[R]$.
\end{lemma}

We now prove the lower bound for the RH-SD and the RH-CD problems.

\begin{theorem}\label{thm:OHSD_lower_bound}
  The RH-SD problem requires $\Omega(n\log n)$ time under the algebraic computation tree model.
\end{theorem}

\begin{proof}
  By reduction from the $\varepsilon$-Closeness problem. Consider an instance of the $\varepsilon$-Closeness problem given by a set $x_1,\ldots,x_n$ and $\varepsilon>0$ of $n+1$ real numbers.
  Using the construction we described above, we create the disjoint sets $R$ and $B$ of $2n+2$ red and $n-1$ blue points.
  We add additional $n+3$ points to $B$ by placing blue points on the line $\ell$ for values of $x$ in the interval $(m,x_{min}-\varepsilon)\cup (x_{max},M)$.
  By \Cref{prop:rch_inclusion}, these additional points are not contained in the rectilinear convex hull of $R$, regardless of the distances between consecutive numbers in the set $x_1,\ldots,x_n$.
  See \Cref{fig:lower_bound_separability_containment:1}.

  We use an algorithm to solve the RH-SD problem on the sets $R$ and $B$ of $2n +2$ red and $2n + 2$ blue points.
  If the algorithm returns true, we reject the instance of the $\varepsilon$-Closeness problem, otherwise we accept the instance.
  By \Cref{lem:construction}, there is at least one blue point contained in the rectilinear convex hull of $R$ if, and only if, there is a pair of consecutive numbers in $x_1,\ldots,x_n$ at distance less than $\varepsilon$.
  Therefore, we have correctly solved the $\varepsilon$-Closeness problem in $O(n)$ time plus the time required to solve the RH-SD problem.
\end{proof}

\begin{figure}[ht]
  \centering

  \subcaptionbox
  {\label{fig:lower_bound_separability_containment:1}
    A blue point is contained in the rectilinear convex hull of $R$ if, and only if, the difference between a pair of consecutive numbers in the set $x_1,\ldots,x_n$ is less than $\varepsilon$.
  }
  [0.45\linewidth][c]{\includegraphics{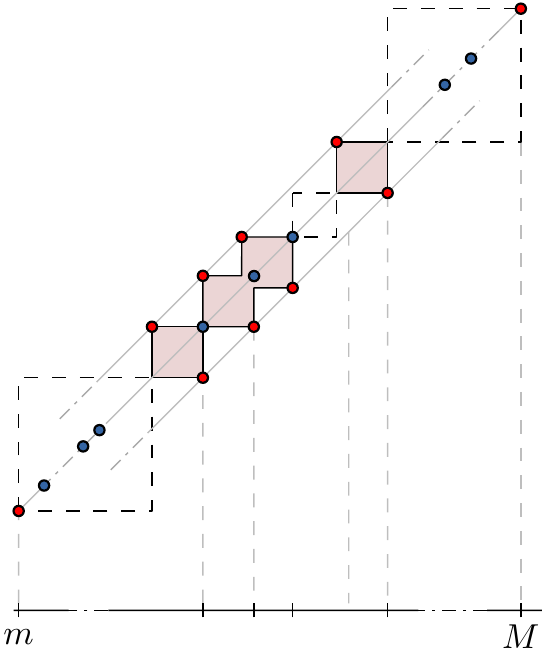}}
  ~
  \subcaptionbox
  {\label{fig:lower_bound_separability_containment:2}
    All the blue points are contained in the rectilinear convex hull of $R$ if, and only if, the maximum difference between a pair of consecutive numbers in the set $x_1,\ldots,x_n$ is less than $\varepsilon$.
  }
  [0.45\linewidth][c]{\includegraphics{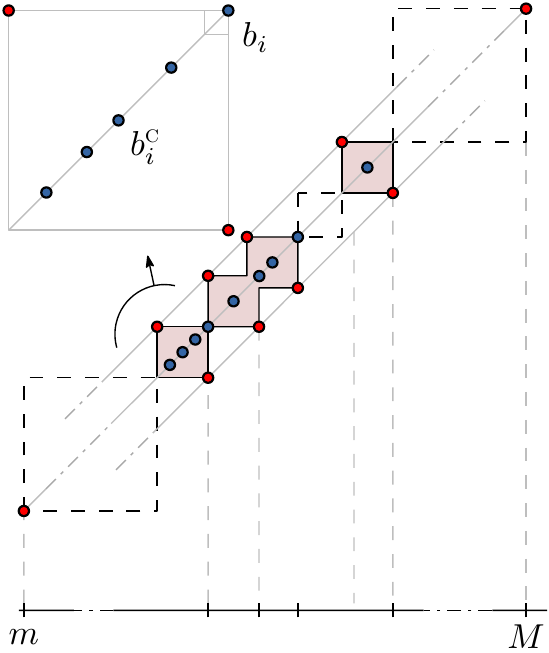}}

  \caption
  {
    Illustrations for \subref{fig:lower_bound_separability_containment:1} \Cref{thm:OHSD_lower_bound} and \subref{fig:lower_bound_separability_containment:2} \Cref{thm:OHCD_lower_bound}.
  }
  \label{fig:lower_bound_separability_containment}
\end{figure}

\begin{theorem}\label{thm:OHCD_lower_bound}
  The RH-CD problem requires $\Omega(n\log n)$ time under the algebraic computation tree model.
\end{theorem}

\begin{proof}
  By reduction from the CGE problem. Consider an instance of the CGE problem given by a set $x_1,\ldots,x_n$ and $\varepsilon>0$ of $n+1$ real numbers.
  Using the construction we described above, we create the disjoint sets $R$ and $B$ of $2n+2$ red and $n-1$ blue points.
  We add additional $n+3$ points to $B$ by placing on the line $\ell$ three blue points for values of $x$ in the interval $(x_{min} -\varepsilon,x_{min})$, and a point $b_i^{\textsc{c}}=(x_i -\frac{\varepsilon}{2}, x_i - \frac{\varepsilon}{2})$ for $1\leq i\leq n$.
  By \Cref{prop:rch_inclusion}, these additional points are contained in the rectilinear convex hull of $R$, regardless of the distances between consecutive numbers in the set $x_1,\ldots,x_n$.
  See \Cref{fig:lower_bound_separability_containment:2}.

  We use an algorithm to solve the RH-CD problem on the sets $R$ and $B$ of $2n+2$ red and $2n+2$ blue points.
  If the algorithm returns true we accept the instance of the CGE problem, otherwise we reject the instance.
  By \Cref{lem:construction}, there is at least one blue point not contained in the rectilinear convex hull of $R$ if, and only if, the maximum distance between consecutive numbers in $x_1,\ldots,x_n$ is at least $\varepsilon$.
  Therefore, we have correctly solved the CGE problem in $O(n)$ time plus the time required to solve the RH-CD problem.
\end{proof}

A solution of the RH-PI problem can be trivially transformed into a solution of the problems RH-SD and RH-CD in $O(1)$ and $O(n)$ time, respectively: An instance of the RH-SD problem is positive if the subset of blue points contained in the rectilinear convex hull of the red point set is empty, whereas an instance of the RH-CD is positive if the subset contains $n$ points.
Hence, as a consequence of \Cref{thm:OHSD_lower_bound,thm:OHCD_lower_bound} we obtain the following theorem.

\begin{theorem}\label{thm:OCHPI_lower_bound}
  The RH-PI problem requires $\Omega(n\log n)$ time in the algebraic computation tree model.
\end{theorem}

\section{Concluding remarks}
\label{sec:concluding_remarks}

We described efficient algorithms to compute the orientations of the lines of $\o$ for which there is an $\o$-convex hull separating $R$ from $B$.
If $\o$ is formed by two lines we considered two cases.
In the first case we simultaneously rotate both lines around the origin.
In the second case we rotate one of the lines while the second one remains fixed.
In both cases our algorithms run in optimal $O(n \log n)$ time and $O(n)$ space.
The optimality is shown by providing a matching lower bound for the problem.
If instead $\o$ is formed by $k\geq2$ lines, we simultaneously rotate all the lines of $\o$ around the origin.
Our algorithm runs in this case in $O(\sfrac{1}{\Theta} \cdot N \log N)$ time and $O(\sfrac{1}{\Theta} \cdot N)$ space, where $N = \max \{ k, \vert R \vert + \vert B \vert \}$ and $\Theta$ is the smallest among the sizes of the $\o$-wedges induced by the set of orientations.

The central strategy of all our algorithms is to perform an angular sweep in which, while we change the orientations of the lines of $\o$, we keep the number of blue points contained in the $\o$-convex hull of $R$.
Note that, without increasing the time and space complexities, the angular sweep can be easily adapted to compute the set of angular intervals for which the $\o$-convex hull of $R$ contains the minimum number of blue points.
Using the terminology from Houle~\cite{houle_1989,houle_1993}, if such number is equal to zero, then the particular $\o$-convex hull is a \emph{strong separator} for $R$ and $B$, otherwise is a \emph{weak separator} for $R$ and $B$.
Hence, our algorithm can be used to solve a variation of the so-called \emph{weak separability problem} in which a given bichromatic point set is separated by an $\o$-convex hull.

Finally, we remark that the sweeping process can be also modified to add further optimizations.
By applying the techniques from \cite{alegria_2018,alegria_2020} for example, we can obtain the $\o$-convex hull with the minimum (or maximum) area, perimeter, or number of vertices, that is either a strong or a weak separator for $R$ and $B$.
These additional optimizations do not increase neither the time nor the space complexities of the original algorithms.


\bibliography{references}

\end{document}